%% file: KaraT10BuchiData.tex
\documentclass[11pt,letterpaper]{article}

\usepackage{latexsym}
\usepackage{amssymb}
\usepackage{amsmath}
\usepackage{algorithm}
\usepackage{algorithmic}

\newtheorem{defn}{Definition}
\newtheorem{rem}[defn]{Remark}

\newtheorem{theorem}[defn]{Theorem}
\newtheorem{lemma}[defn]{Lemma}
\newtheorem{proposition}[defn]{Proposition}

\newenvironment{definition}{\defn \em }
 {\vspace{0.5 em}} 

\newenvironment{remark}{\rem \em}{\vspace{0.5 em}}

\newenvironment{example}{\eg \em}{\vspace{0.5 em}}

\newenvironment{proof}
{\vspace{0.5 em} \noindent {\bf Proof.}}{\QED {\vspace{0.7 em}}}

\def\squareforqed{$\Box$}
\def\QED{\ifmmode\squareforqed\else{\unskip\nobreak\hfil
\penalty50\hskip1em\null\nobreak\hfil\squareforqed
\parfillskip = 0pt\finalhyphendemerits = 0\endgraf}\fi}

\include{commands}

\begin{document}

\title{Extending B\"uchi Automata
\\
with Constraints on Data Values\thanks{We acknowledge the financial support by the European FET-Open Project FoX (grant agreement 233599) and the German DFG (grant SCHW 678/4-1).
}}

\author{
Ahmet Kara  
\\
Technical University of Dortmund   
\\
Email: \texttt{ahmet.kara@cs.tu-dortmund.de}
\\
\\
Tony Tan
\\
University of Edinburgh
\\
Email: \texttt{ttan@inf.ed.ac.uk}
}

\date{}

\maketitle

\begin{abstract}
Recently data trees and data words have received
considerable amount of attention
in connection with XML reasoning and system verification.
These are trees or words that,
in addition to labels from a finite alphabet,
carry data values from an infinite alphabet (data).
In general it is rather hard to obtain logics for data words and trees
that are sufficiently expressive, but still have reasonable complexity
for the satisfiability problem.
In this paper we extend and study the notion of B\"uchi automata for $\omega$-words with data.
We prove that the emptiness problem for such extension
is decidable in elementary complexity.
We then apply our result to show the decidability
of two kinds of logics for $\omega$-words with data:
the two-variable fragment of first-order logic
and some extensions of classical linear temporal logic 
for $\omega$-words with data.
\end{abstract}

\section{Introduction}
\label{s: intro}

The classical theory of automata and formal languages deals primarily with
languages over finite alphabets.
A natural extension of formal languages, regular or contex-free,
is one that permits the alphabet to be infinite~\cite{boasson,luc-lics06,luc-pods06-jacm,cheng-kaminski-cfg,kaminski-francez,kaminski-tan,NSV-mfcs}.
Most of the extensions, however, lack the usual nice decidability properties of
automata over finite alphabets, unless strong restrictions are imposed.

Recently the subject of languages over infinite alphabets
received much attention due to its connection with XML reasoning and system specification.
The most natural model for XML documents is label unranked trees,
in which each node has a label from a finite alphabet.
Thus, standard technique in automata theory can be applied~\cite{libkin-unranked-tree,neven-csl,schwentick-jcss}.
However, real XML documents carry {\em data}, which usually come from an infinite set,
and it is essential to reason about those data values.
Thus, there is a need to look for {\em decidable formalism}
in the presence of a second, infinite alphabet. 

A similar scenario may happen in system specification where
$\omega$-words (words of infinite length) are used to describe system behaviors.
In this case a position in the word represents a point in time,
while the label of the position indicates the atomic propositions that hold at that time.
The number of atomic propositions is usually only finitely many,
and thus, can be encoded as finite alphabets.
The most common tool for reasoning with $\omega$-word is arguably B\"uchi automata, 
due to its expressiveness and
the low complexities for its standard decision problems.
For example, it captures the so-called monadic second order (MSO) logic, 
and hence the specification languages such as 
Linear Temporal Logic (LTL) and $\mu$-calculus.
However, the behaviour of many systems includes
properties that cannot be captured by finite alphabets.
A typical example is reasoning about the contents of variables,
that store values from the infinite domains like the integers or strings.
Thus, it is also natural to look for some formalisms
that allow us to reason about $\omega$-words with data values
that come from an infinite domain.

Our focus in this paper is {\em data $\omega$-word}, 
that is, $\omega$-words in which each position also carries a data value from an infinite alphabet.
Looking at the literature~\cite{luc-pods06-jacm,luc-lics06,therien-and-co,fo2-lpar,demri-lfcs07,
demri-lazic-tocl,kaminski-francez,kaminski-tan,neven-csl,NSV-mfcs}
one can immediately notice that decidable formalisms for data $\omega$-words are hard to obtain,
unless strong restrictions are imposed.
Nevertheless, some significant progress have been made 
recently~\cite{luc-lics06,demri-lazic-tocl,demri-lfcs07}.
A deep result in~\cite{luc-lics06} shows that the restriction of 
first-order logic to its two variable fragment, $\sfFO^2$,
remains decidable over data $\omega$-words.
The pioneering works in Linear Temporal Logic for $\omega$-words with data
are the papers~\cite{demri-lazic-tocl,demri-lfcs07}.
In~\cite{demri-lfcs07} an extension of Linear Temporal Logic (LTL) to handle data values
is proposed and its satisfiability problem is shown to be decidable.
In papers~\cite{luc-lics06,demri-lfcs07} the satisfiability problem, even though is decidable,
has unknown upper bound complexity.
The decidability is obtained by reducing the satisfiability problem to 
the reachability problem in Petri nets,
the precise complexity of which has been open for many years,
though it is known to be in $\sfEXPSPACE$-hard.
In the paper~\cite{demri-lazic-tocl} the logic is decidable, but not primitive recursive,
for finite data words, while it becomes undecidable for $\omega$-words.
The paper~\cite{demri-lfcs07} also contains
a logic which is decidable in $\sfPSPACE$. 
However, the logic has quite limited expressive power, in which
the finite alphabet for the labels consists of only one single symbol.

In this paper we propose and study an extension of B\"uchi automata
with a formalism to specify constraints on data values.
Roughly those constraints are database theory inspired,
called {\em key}-, {\em inclusion}- and {\em denial}-constraints.
A key-constraint states that no two positions labeled with
the same symbol $a$ has the same data value;
inclusion-constraint states that every data value found
in a position with label $a$ is found in a position with label $b$;
while denial-constraint states that
the sets of data values found in positions with labels $a$ and $b$
are disjoint.
Those constraints are very common in database theory.
We show that the emptiness problem for such extension is decidable
in $\sfNEXPTIME$,
whereas if there is no key-constraint,
then the complexity drops to $\sfNP$.
We then apply our results to show the decidability
of two kinds of logics for data $\omega$-words:
the two-variable fragment of first-order logic
and some extensions of classical linear temporal logic for data $\omega$-words.
Both have elementary complexity.

The vocabulary for the two-variable logic that we consider
here has only  the {\em successor} relation on the positions in the $\omega$-word
and the {\em data equality},
in addition to the finite number of unary predicates for the finite labeling.
In~\cite{luc-lics06} the vocabulary includes
the {\em order} on the positions in the $\omega$-words
and as mentioned earlier, the satisfiability problem for the two-varible logic
becomes at least as hard as the reachability problem for Petri nets.

Another work that is related to our work is the remarkable result in~\cite{luc-pods06-jacm},
which shows that 
for two-variable fragment of first-order logic over {\em finite unranked data trees},
with vocabulary consists of successor and data equality,
is decidable in 3-$\sfNEXPTIME$.
Another proof with different approach for the restricted case of {\em finite data words}
was later obtained in~\cite{fo2-lpar}.

The paper is organized as follows.
In Section~\ref{s: notations} we define the notations and tools
that we are going to use in this paper. 
In Section~\ref{s: diamond aut} we introduce the extension of B\"uchi automata
by equipping it with data-constraints and we prove that
the emptiness problem is decidable in elementary complexity.
We call this model {\em B\"uchi automata with data-constraints} (ADC).
In Section~\ref{s: diamond aut with profiles}
we further extend ADC with operators for comparing the equality between neighboring data values,
which we call {\em profile B\"uchi automata with data-constraints}.
The emptiness problem for this model is also decidable in elementary complexity.
Then in Section~\ref{s: fo2} we present a decision procedure
for the satisfiability problem  of the two-variable fragment of first-order logic.
Finally in Section~\ref{s: LTL data value}
we introduce a version of Linear Temporal Logic (LTL)
that is equipped with some operators for data value comparisons.
For this also we prove that the satisfiability problem
is decidable in elementary complexity. 

\paragraph*{Acknowledgement}
We thank Claire David, Leonid Libkin and Thomas Schwentick for fruitful discussions.

\section{Notations}
\label{s: notations}

\subsection{Data words}
\label{ss: data words}

Let $\Sigma$ be a finite alphabet
and $\fD$ an infinite set of data values.
A {\em finite} word is an element of $\Sigma^{\ast}$,
while an $\omega$-word is an element of $\Sigma^{\omega}$.
A finite {\em data word} is an element of $(\Sigma \times \fD)^{\ast}$,
while a {\em data $\omega$-word} is an element of $(\Sigma \times \fD)^{\omega}$.

We write a data (finite or $\omega$-) word $w$ as 
${a_1 \choose d_1}{a_2 \choose d_2}\cdots $,
where $a_1,a_2,\ldots \in \Sigma$ and $d_1,d_2,\ldots \in \fD$.
The symbol $a_i$ is the label of position $i$,
while the value $d_i$ is the data value in position $i$.
The projection of $w$ to the alphabet $\Sigma$ is denoted by 
$\sfProj(w) = a_1 a_2 \cdots$.
A position in $w$ is called an $a$-position,
if the label is $a$.
We denote by $V_w(a)$, the set of data values
found in $a$-positions in $w$, i.e., 
$V_w(a) = \{d_i \mid a_i = a\}$, for each $a \in \Sigma$.
Note that some $V_w(a)$'s may be infinite,
while some others finite.

\subsection{Data-constraints: constraints on the data values}
\label{ss: constraint}

There are three kinds of data-constraints over the alphabet $\Sigma$:
\begin{enumerate}
\item
{\em key-constraints}, written in the form:
$V(a) \mapsto a$, where $a \in \Sigma$.
\item
{\em inclusion-constraints}, written in the form:
$V(a) \subseteq \bigcup_{b \in R} V(b)$,
where $a \in \Sigma$, $R \subseteq \Sigma$.
\item
{\em denial-constraints}, written in the form:
$V(a)\cap V(b) = \emptyset$,
where $a,b \in \Sigma$.
\end{enumerate}
Whether a data word $w$ satisfies a data-constraint $C$, written  as $w \models C$,
is defined as follows.
\begin{enumerate}
\item
$w\models V(a)\mapsto a$,
if every two $a$-positions in $w$ have different data values.
\item
$w \models V(a)\subseteq \bigcup_{b\in R} V(b)$,
if $V_w(a) \subseteq \bigcup_{b \in R} V_w(b)$.
\item
$w\models V(a)\cap V(b) = \emptyset$,
if $V_w(a) \cap V_w(b) = \emptyset$.
\end{enumerate}
If $\CC$ is a collection of data-constraints,
then we write $w \models \CC$, if $w\models C$ for all $C \in \CC$.

\subsection{Transition systems and B\"uchi automata}
\label{ss: buchi}

A transition system over the alphabet $\Sigma$ 
is a tuple $\MM = \langle Q, \mu \rangle$,
where $Q$ is a finite set of states and $\mu \subseteq Q \times \Sigma \times Q$
is the set of transitions.

A B\"uchi automaton $\AA$ over the alphabet $\Sigma$ is simply
a transition system $\MM$ with a designated initial state $q_0$
and a set $F\subseteq Q$ of final states.
In such case, we write $\AA = \MM_{q_0}^{F}$ and
the system $\MM$ is called the {\em transition system} of $\AA$.

A run of $\AA$ on an $\omega$-word $w = a_1a_2\cdots$
is a sequence $\rho = p_1 p_2 \cdots$ of states in $Q$
such that $(q_0,a_1,p_1) \in \mu$ and $(p_i,a_{i+1},p_{i+1}) \in \mu$, for each $i=1,2,\ldots$.
Note that we exclude the initial state in the run $\rho$.
This is done for our convenience of indexing.

Let $\sfInf(\rho)$ denote the set of states that appear infinitely many times in $\rho$.
The run $\rho$ is accepting, if $\sfInf(\rho) \cap F \neq \emptyset$.
An $\omega$-word $w \in \LL(\AA)$,
if there exists an accepting run of $\AA$ on $w$.
As usual, $\LL(\AA)$ denotes the set of $\omega$-words accepted by 
the automaton $\AA$.

\subsection{Presburger automata}
\label{ss: tools}

\subsubsection*{Existential Presburger formula}

Atomic Presburger formulae are of the form:
$x_1 + x_2 + \cdots x_n \leq y_1 + \cdots + y_m$, or
$x_1 + \cdots x_n \leq K$, or $x_1 + \cdots x_n \geq K$,
for some constant $K \in \nn$.  {\em Existential Presburger formulae}
are Presburger formulae of the form 
$\exists \bar x \ \phi$, 
where $\phi$ is a Boolean combination of atomic Presburger formulae. 

We will be using Presburger formulae defining Parikh images of words.
Let $\Sigma=\{a_1,\ldots,a_k\}$ be a finite alphabet, and 
let $v \in \Sigma^*$ be a finite word.
We denote by $\#_{v}(a_i)$ the number of occurrences of $a_i$ in $v$. 
By $\sfParikh(v)$ we mean the {\em Parikh image} of $v$, i.e.,
$(\#_v(a_1),\ldots,\#_v(a_k))$.

With alphabet letters $a_1,\ldots,a_k$, we associate variables $x_{a_1},\ldots,x_{a_k}$.
A Presburger formula $\varphi$ with free variables $x_{a_1},\ldots,x_{a_k}$
is said to be a formula over the alphabet $\{a_1,\ldots,a_k\}$.
A word $v \in \Sigma^{\ast}$ satisfies it, 
written as $v \models \varphi(x_{a_1},\ldots,x_{a_k})$ if and only if $\varphi(\sfParikh(v))$
holds. 

\subsubsection*{Presburger automata}

A {\em Presburger automaton} is a pair $(\AA_{\tsffin},\varphi)$, 
where $\AA_{\tsffin}$ is a finite state automaton for {\em finite words} and
$\varphi(x_{a_1},\ldots,x_{a_k})$ is an existential Presburger formula over the alphabet $\Sigma$.  
A word $w$ is accepted by $(\AA_{\tsffin},\varphi)$, denoted by $\LL(\AA_{\tsffin},\varphi)$, if 
$w \in \LL(\AA_{\tsffin})$ and $\varphi(\sfParikh(w))$ holds.

Note that as convention, we will use the symbol $\AA_{\tsffin}$
for finite state automata that works over {\em finite} words.
We reserve the symbol $\AA$ for B\"uchi automata,
which works over $\omega$-words.

As in~\cite{luc-pods06-jacm,fo2-lpar},
the following result is the basis for all
the decidability results in this paper.

\begin{theorem}{\cite{schwentick-icalp04}}
\label{t: emptiness presburger}
The emptiness problem for presburger automata is decidable in $\sfNP$.
\end{theorem}

\vspace{-0.3 cm}

\section{Automata with data-constraints}
\label{s: diamond aut}

In this section we extend the definition of B\"uchi automata
with data-constraints over the input alphabet $\Sigma$.
We then provide a decision procedure for its emptiness problem,
from which all other decision procedures in this paper are extended.

\begin{definition}
\label{d: diamond aut}
An {\em Automaton with Data-constraints}, or in short ADC,
is a pair $(\AA,\CC)$, where
$\AA$ is a B\"uchi automaton and
$\CC$ is a collection of data-constraints
over the alphabet $\Sigma$.
\end{definition}

\noindent
Let $w = {a_1 \choose d_1} {a_2 \choose d_2}\cdots$ be a data $\omega$-word.
The ADC $(\AA,\CC)$ accepts the data $\omega$-word $w$, if 
$\sfProj(w) \in \AA$ and $w\models \CC$.
We denote by $\LL(\AA,\CC)$ the language that
consists of all the data $\omega$-words accepted 
by the ADC $(\AA,\CC)$.

We consider the following problem.
\begin{center}
\fbox{
\begin{tabular}{ll}
{\sc Problem}: & {\sc Omega-SAT-ADC} \\ \hline
{\sc Input}: & An automaton with data-constraints $(\AA,\CC)$ 
\\
{\sc Question}: & Is there an data $\omega$-word $w \in \LL(\AA,\CC)$?
\end{tabular}
}
\end{center}

\begin{theorem}
\label{t: emptiness diamond aut}
The problem {\sc SAT-ADC} is decidable in $\sfNEXPTIME$.
Moreover, if the collection $\CC$ of data constraints does not contain
key-constraints, then it is decidable in $\sfNP$.
\end{theorem}

For the proof we first introduce some essential notations
in Subsection~\ref{ss: notation proof one},
then we outline the $\sfNEXPTIME$ algorithm in Subsection~\ref{ss: algorithm diamond aut}.
The $\sfNP$ algorithm can be found in Appendix~\ref{app: s: proof np emptiness diamond aut}.

 Before we start the first proof in this paper,
we want to remark the similarities and differences 
between the technique in this paper and the one in~\cite{luc-lics06}.
 The only similarity is that all techniques rely quite heavily on Presburger counting.
However, there is a different emphasis in the counting process:
in~\cite{luc-pods06-jacm} the technique is to count the number of the so called 
{\em dog} labels and {\em sheep} labels (see pp. 35--36 in~\cite{luc-pods06-jacm}),
where intuitively, the dog labels are used to represent the data values.
In this paper the technique involves counting directly the ``number'' of data values.


\subsection{Some notations for the proof of Theorem~\ref{t: emptiness diamond aut}}
\label{ss: notation proof one}

For a data $\omega$-word $w$ and a non-empty subset $S \subseteq \Sigma$, 
we denote by 
$$
[S]_{w} = \bigcap_{a \in S} V_{w}(a) \cap
\bigcap_{b \notin S} \overline{V_{w}(b)},
$$
where $\overline{V_{w}(b)}$ denotes the complement of $V_w(b)$,
i.e. $\fD - V_{w}(b)$.
It must be noted that the sets $[S]_{w}$'s are disjoint,
and for each $a \in \Sigma$,
$V_{w}(a)$ is partitioned into 
$V_{w}(a) = \bigcup_{a \in S} [S]_{w}$.
These two properties (disjointness and partition) of $[S]_w$'s
are very crucial in our decision procedure.

According to the cardinalities of $[S]_{w}$'s,
we divide the non-empty subsets $S \subseteq \Sigma$ into three classes:
\begin{itemize}
\item
$\SS_{0}(w) = \{S \mid [S]_{w} = \emptyset\}$.
\item
$\SS_{\tsffin}(w) = \{S \mid [S]_{w} \ \mbox{is a finite non-empty set}  \}$.
\item
$\SS_{\infty}(w) = \{S \mid [S]_{w} \ \mbox{is an infinite set}  \}$.
\end{itemize}

\begin{proposition}~\cite[Proposition~1]{fo2-lpar}
\label{p: data-constraint}
For every data $\omega$-word $w$, the following holds.
\begin{enumerate}
\item 
$w\models V(a) \subseteq \bigcup_{b \in R} V(b)$
if and only if
$S \in \SS_0(w)$, for all $S$ such that $a \in S$, but $S \cap R = \emptyset$.
\item
$w\models V(a) \cap V(b) = \emptyset$
if and only if
$S \in \SS_0(w)$, for all $S$ such that $a,b \in S$.
\end{enumerate}
\end{proposition}
\begin{proof}
(2) is immediate from the definition of $[S]_w$, while
(1) follows from the fact that
\begin{eqnarray*}
V_w(a) \subseteq \bigcup_{b \in R} V_w(b)
& \mbox{if and only if} &
V_w(a) \cap \bigcap_{b \in R} \overline{V_w(b)} = \emptyset.
\end{eqnarray*}
\end{proof}

\subsection{The algorithm}
\label{ss: algorithm diamond aut}

\noindent
Let $(\AA,\CC)$ be the given ADC and 
$\MM = \langle Q, \mu \rangle$ be the transition system,
where $\AA=\MM_{q_0}^{F}$.
Roughly our algorithm to determine whether $\LL(\AA,\CC) = \emptyset$
is as follows.
\begin{enumerate}
\item
Guess a partition $\SS_0,\SS_{\tsffin},\SS_{\infty}$ of the sets $2^{\Sigma} - \{\emptyset\}$
that respects the following conditions.
\begin{itemize}
\item[(C1)]
If the inclusion-constraint $V(a) \subseteq \bigcup_{b\in R} V(b)$ is in $\CC$,
then all the sets $S$, where $a\in S$ and $S \cap R =\emptyset$, are in $\SS_0$.
\item[(C2)]
If the denial-constraint $V(a)\cap V(b) = \emptyset$ is in $\CC$,
then all the sets $S$, which contains both $a$ and $b$, are in $\SS_0$.
\end{itemize}
The intended meaning of the guesses $\SS_0,\SS_{\tsffin}, \SS_{\infty}$ are
the sets $\SS_0(w)$, $\SS_{\tsffin}(w)$ $\SS_{\infty}(w)$, respectively,
for some $w \in \LL(\AA,\CC)$.
\\
Moreover, Conditions (C1) and (C2) must be respected due to Proposition~\ref{p: data-constraint}.
\item
Construct the following two items,
of which the details are provided below.
\begin{enumerate}
\item[$(a)$]
A new alphabet $\tilde{\Sigma}$,
which depend on the original alphabet $\Sigma$ and the sets in $\SS_{\infty}$
\item[$(b)$]
A transition system $\tilde{\MM} =\langle \tilde{Q}, \tilde{\mu}\rangle$
over the alphabet $\tilde{\Sigma}$,
which depends on the original transition system $\MM$
and the sets in $\SS_{\infty}$.
\end{enumerate}
\item
Non-deterministically choose one state $q \in \tilde{Q}$
and construct the following two items.
\begin{itemize}
\item
a Presburger automaton $(\tilde{\AA}_{\tsffin},\varphi)$, where 
$\tilde{\AA}_{\tsffin} = \tilde{\MM}_{q_0}^{\{q\}}$ and
the formula $\varphi$ depends on the partition $\SS_0,\SS_{\tsffin},\SS_{\infty}$ 
and the constraints in $\CC$;
\item
a B\"uchi automaton $\tilde{\AA}$, which depends on the constraints in $\CC$,
the new transition system $\tilde{\MM}$, and the sets in $\SS_{\infty}$.
\end{itemize}
\item
Test the emptiness of $\LL(\tilde{\AA}_{\tsffin},\varphi)$ and $\LL(\tilde{\AA})$.
\\
Then, $\LL(\AA,\CC)\neq \emptyset$ if and only if $\LL(\tilde{\AA}_{\tsffin},\varphi) \neq \emptyset$
and $\LL(\tilde{\AA})\neq \emptyset$.
\end{enumerate}
In the paragraphs below we will outline the details of Steps~(2) and~(3).
The analysis of the complexity is given in Appendix~\ref{app: s: complexity analysis}.

The proof of the correctness will follow from our claim that
$\LL(\AA,\CC) \neq \emptyset$ if and only if
there exist some ``correct'' guesses for $\SS_0,\SS_{\tsffin},\SS_{\infty}$ in Step~(1)
and the state $q \in \tilde{Q}$ in Step~(3) such that
$\LL(\tilde{\AA}_{\tsffin},\varphi) \neq \emptyset$
and $\LL(\tilde{\AA})\neq \emptyset$.
The details of the proof of the correctness will be given in Appendix~\ref{app: s: proof correctness}.
The main idea of the proof is that 
from a word $u \in\LL(\tilde{\AA}_{\tsffin},\varphi)$,
we can construct a finite data word $w$,
and from an omega word $v \in \LL(\tilde{\AA})$,
we can construct a data $\omega$-word $w'$
such that $ww' \in \LL(\AA,\CC)$.

\subsubsection*{Constructing the alphabet $\tilde{\Sigma}$ and the transition system $\tilde{\MM}$}

We define a set $\Sigma(\SS_{\infty}) = \{(a,S) \mid a \in S \ \mbox{and} \ S \in \SS_{\infty}\}$. 
Then, the new alphabet $\tilde{\Sigma}$ is
$\tilde{\Sigma} = \Sigma \cup \Sigma(\SS_{\infty})$.
The transition system $\tilde{\MM} = \langle \tilde{Q}, \tilde{\mu}\rangle$ 
is defined as $\tilde{Q} = Q$ and 
$\tilde{\mu} = \mu \cup
\{(p,(a,S),q) \mid (p,a,q) \in \mu \ \mbox{and} \ (a,S)\in \Sigma(\SS_{\infty})\}$.

\subsubsection*{Constructing the Presburger automaton $(\tilde{\AA\;}_{\!\!\tsffin},\varphi)$}

Let $q \in \tilde{Q}$ be the state chosen non-deterministically in Step~(3).
The automaton $\tilde{\AA_{\tsffin}}$ is simply $\MM_{q_0}^{\{q\}}$.
The Presburger formula $\varphi$ is defined as follows.
Let $S_1,\ldots,S_{m}$ be the enumeration of 
non-empty subsets of $\Sigma$, where $m = 2^{|\Sigma|}-1$.  

The formula $\varphi$ is of the form 
$\exists z_{S_1} \; \cdots \; \exists z_{S_m}\ \psi$, 
where $\psi$ is the following quantifier-free formula:
$$
\bigwedge_{a \in \Sigma}
x_{a} \geq \sum_{S \ni a} z_S
\quad
\wedge
\quad
\bigwedge_{S \in \sSS_0 \cup \sSS_{\infty}} z_S = 0 
$$
$$
\quad
\wedge
\quad
$$
$$
\bigwedge_{S \in \sSS_{\tsffin}} z_S \geq 1
\quad
\wedge
\quad
\bigwedge_{V(a)\mapsto a \in \sCC} 
x_{a} = \sum_{a \in S} z_S 
$$
Note the constructed formula $\varphi$
does not involve the symbols in $\Sigma(\SS_{\infty})$.

\subsubsection*{Constructing the B\"uchi automaton $\tilde{\AA}$}

The B\"uchi automaton $\tilde{\AA}$ is
simply the intersection of $\tilde{\MM}_{q}^{F}$
with the automaton that checks the following conditions.
\begin{enumerate}
\item
Each $(a,S) \in \Sigma(\SS_{\infty})$
appears infinitely many times.
\item
If the key-constraint $V(a) \mapsto a \in \CC$,
then the symbol $a$ does not appear. 
\end{enumerate}

\section{Automata with data-constraints and profiles}
\label{s: diamond aut with profiles}

Given a data word $w = {a_1 \choose d_1} {a_2 \choose d_2} \cdots$, 
the {\em profile word} of $w$, denoted by $\sfProfile(w)$, 
is the word 
$$
\sfProfile(w) = (a_1,(L_1,R_1)),(a_2,(L_2,R_2)),\ldots
\in (\Sigma \times \{\ast,\top,\bot\}\times \{\ast,\top,\bot\})^{\omega}
$$ 
such that for each position $i=1,2,\ldots$, 
the values of $L_i$ and $R_i$ are either $\top$, or $\bot$, or $\ast$.  
If $L_i=\top$ and $i>1$,
it means that the position on the left, $i-1$, has the same data value
as position $i$; otherwise $L_i=\bot$. 
If $i=1$ (i.e., there is no position on the left), then $L_i=\ast$. 
The meaning of the $R_i$'s is
similar with respect to positions on the right of $i$.

A {\em profile B\"uchi automaton} $\AA$ is a B\"uchi automaton over the
alphabet $\Sigma \times \{\ast,\top,\bot\} \times \{\ast,\top,\bot\}$.
It defines a set $\LL_{data}(\AA)$ of data words as follows: 
$w \in \LL_{data}(\AA)$ if and only if $\AA$ accepts $\sfProfile(w)$ in
the standard sense.

A profile B\"uchi automaton with data-constraints is a tuple $(\AA,\CC)$,
where $\AA$ is a profile B\"uchi automaton 
and $\CC$ is a collection of data-constraints.
It defines a set of data $\omega$-words as follows.
An data $\omega$-word $w$ is accepted by $(\AA,\CC)$ if 
$\sfProfile(w) \in \LL(\AA)$
and $w\models \CC$.

\begin{theorem}
\label{t: emptiness profile diamond aut}
The emptiness problem for profile B\"uchi automata with data-constraints is in 2-$\sfNEXPTIME$.
\end{theorem}

We give a sketch of the proof in Subsection~\ref{ss: sketch proof}.
The details can be found in Appendix~\ref{app: s: proof emptiness profile diamond aut}.
Before that we give a slight extension of profile B\"uchi automata with data-constraints,
which we call {\em profile B\"uchi automata with data-constraints on the state alphabet}.
It is a trivial extension, 
but it will be very useful for our presentation in Appendix~\ref{app: s: proof complexity LTL}.

\subsection{Profile B\"uchi automata with data-constraints on the state alphabet}
\label{s: general profile buchi automata}

\begin{definition}
A profile B\"uchi automaton with {\em data-constraints on the state alphabet}
is a pair $(\AA,\CC)$, where
\begin{itemize}
\item
$\AA = \langle Q, q_0, \mu, F\rangle$ is a profile B\"uchi automaton, and
\item
$\CC$ is a collection of data-constraints over the {\em state} alphabet $Q$
(instead of over the input alphabet $\Sigma$ as in Definition~\ref{d: diamond aut}).
\end{itemize}
\end{definition}

\noindent
Let $w = {a_1 \choose d_1} {a_2 \choose d_2}\cdots$ be an data $\omega$-word,
and $\rho = p_1 p_2 \cdots$ be a run of $\AA$ on $\sfProfile(w)$.
The {\em induced data word} of $w$ on $\rho$ is the data word
$\rho(w) = {p_1 \choose d_1} {p_2 \choose d_2} \cdots$.

The automaton $(\AA,\CC)$ accepts the data $\omega$-word $w$, if 
there is an accepting run $\rho$ of $\AA$ on $\sfProj(w)$
such that $\rho(w) \models \CC$.
We denote by $\LL(\AA,\CC)$ the language that
consists of all the data $\omega$-words accepted 
by the automaton $(\AA,\CC)$.

The upper bound in Theorem~\ref{t: emptiness diamond aut} still holds 
for the emptiness problem of this type of automaton.
Indeed given an input $(\AA,\CC)$,
a profile B\"uchi automaton with data-constraints on state alphabet $Q$,
we can reduce it to $(\AA',\CC')$, a profile B\"uchi automaton
with data-constraints over the alphabet $Q \times \Sigma$ as follows.
The automaton $\AA'$ accepts the $\omega$-word (with profiles) 
$(q_1,a_1,\sfprofile_1)(q_2,a_2,\sfprofile_2)\cdots$
if and only if
$q_1 q_2\cdots$ is an accepting run of
the automaton $\AA$ on $(a_1,\sfprofile_1)(a_2,\sfprofile_2)\cdots$.
The automaton $\AA'$ simply checks whether
$(q_i,(a_{i},\sfprofile_i),q_{i+1})$ is a valid transition in $\AA$.
Furthermore, the data-constraints over the alphabet $Q$ can be reduced 
to data-constraints over the alphabet $(Q\times \Sigma)$ as follows.
\begin{enumerate}
\item
The key-constraint $V(q) \mapsto q$
can be reduced to $V(q,a)\mapsto (q,a)$, for each $a\in\Sigma$
and denial-constraints $V(q,a)\cap V(q,b)$, whenever $a\neq b$ and $a,b\in \Sigma$.
\item
The inclusion-constraint $V(q) \subseteq \bigcup_{p\in R} V(p)$
can be reduced to inclusion-constraints $V(q,a) \subseteq \bigcup_{p\in R, b\in \Sigma} V(p,b)$,
for each $a\in \Sigma$.
\item
The denial-constraint $V(q) \cap V(p) = \emptyset$
can be reduced to denial-constraints $V(q,a) \cap V(p,b) = \emptyset$,
for each $a,b \in \Sigma$.
\end{enumerate}

\subsection{Sketch of proof of Theorem~\ref{t: emptiness profile diamond aut}}
\label{ss: sketch proof}

The proof is an extension of the one in the previous section.
However, we need a bit more auxiliary terms.  
Let $w = {a_1 \choose d_1}{a_2 \choose d_2}\cdots $ be an data $\omega$-word over $\Sigma$.  
A {\em zone} is a maximal interval $[i,j]$ with the same data values, i.e.
$d_{i}=d_{i+1}=\cdots = d_j$ and $d_{i-1}\neq d_i$ (if $i > 1$) and
$d_{j} \neq d_{j+1}$ (if $j < n$).  
The zone $[i,j]$ is called an $S$-zone, if
$S$ is the set of labels occuring in the zone.

The {\em zonal partition} of $w$ is a sequence $(k_1,k_2,\ldots)$, where
$1\leq k_1 < k_2 < \cdots$ such that
$[1,k_1],[k_1+1,k_2],\ldots$ are the zones in $w$.  
Let the zone $[1,k_1]$ be an $S_1$-zone, $[k_1+1,k_2]$ an $S_2$-zone,
$[k_2+1..k_3]$ an $S_3$-zone, and so on.  
The {\em zonal word} of $w$ is a data word over 
$\Sigma \cup 2^{\Sigma}$ defined as follows.
$$
\sfZonal(w) = 
{S_1 \choose d_{k_1}} a_1 \cdots a_{k_1}  
{S_2 \choose d_{k_2}} a_{k_1+1} \cdots a_{k_2}
\cdots
$$ 
That is, the zonal word of a data word is a word in which each zone is
preceded by a label $S \in 2^{\Sigma}$, if the zone is an $S$-zone.

Moreover, it is sufficient to assume that only the positions labeled
with symbols from $2^{\Sigma}$ carry data values, i.e., the data values of
their respective zones.  
Obviously each two consecutive zones
have different data values, thus,
two consecutive positions (in $\sfZonal(w)$) labeled with
symbols from $2^{\Sigma}$ also have different data values.

Furthermore, if $w$ is a data $\omega$-word over $\Sigma$, then for each $a\in
\Sigma$, $V_w(a) = \bigcup_{a \in S} V_{\tsfZonal(w)}(S)$.
Proposition~\ref{p: zonal disj const} below shows that 
data-constraints for data words over the alphabet $\Sigma$ can be converted into 
data-constraints for the zonal data words over the alphabet $\Sigma \cup 2^{\Sigma}$.
\begin{proposition}
\label{p: zonal disj const}
For every data word $w$ over $\Sigma$,
the following holds.
\begin{itemize}
\item
A data $\omega$-word $w$ satisfies a key-constraint $V(a) \mapsto a$
if and only if its zonal data word $\sfZonal(w)$ satisfies the following constraints.
\begin{itemize}
\item[K1.]
The key-constraints $V(R) \mapsto R$,
for each $R$ such that $a \in R$.
\item[K2.]
The denial-constraints $V(R) \cap V(R') \neq \emptyset$,
for each $R,R'$ such that $a \in R,R'$ and $R\neq R'$.
\item[K3.]
The symbol $a$ occurs at most once in every zone in $\sfZonal(w)$.
\\
(By a zone in $\sfZonal(w)$, we mean a maximal interval
in which every positions are labeled with symbols from $\Sigma$.)
\end{itemize}
\item
A data $\omega$-word $w$ satisfies an inclusion-constraint $V(a) \mapsto \bigcup_{b\in S} V(b)$
if and only if its zonal data word $\sfZonal(w)$ satisfies the following 
inclusion-constraints:
$$
V(R) \subseteq \bigcup_{S'\cap S \neq \emptyset} V(S')
$$
for each $R$ such that $a \in R$.
\item
A data $\omega$-word $w$ satisfies a denial-constraint $V(a) \cap V(b) = \emptyset$
if and only if its zonal data word $\sfZonal(w)$ satisfies the following 
denial-constraints:
$$
V(R) \cap V(R')
$$
for each $R$ and $R'$ such that $a \in R$ and $b\in R'$.
\end{itemize}
\end{proposition}
\begin{proof}
The proof is straightforward due to the fact that
$$
V_w(a) = \bigcup_{a \in S} V_{\tsfZonal(w)}(S).
$$
\end{proof}

Now, given a profile automaton $\AA$ over the alphabet $\Sigma$,
we can construct in exponential time an automaton $\AA^{\tscZonal}$ such that
for all data $\omega$-word $w$,
$$
\sfProfile(w) \in \LL(\AA)
\ \mbox{if and only if} \
\sfProj(\sfZonal(w)) \in \LL(\AA^{\tscZonal}).
$$
Such an automaton $\AA^{\tscZonal}$ is called a {\em zonal automaton} of
$\AA$.  Moreover, if the key-constraint $V(a) \mapsto a \in \CC$, 
we can impose the condition $K3$ in Proposition~\ref{p: zonal disj const}
inside the automaton $\AA^{\tscZonal}$.
This, together with Proposition~\ref{p: zonal disj const}, implies that the
emptiness problem of profile B\"uchi automata with data-constraints 
can be reduced to an instance of the following problem.
\begin{center}
\fbox{
\begin{tabular}{ll}
{\sc Problem}: & {\sc Omega-SAT-zonal-automata} \\ \hline
{\sc Input}: & $\bullet$ a zonal automaton $\AA^{\tscZonal}$ 
			\\
			& $\bullet$ a collection $\CC^{\tscZonal}$ 
			 of data-constraints over the alphabet $2^{\Sigma}$
				\\
{\sc Question}: & is there a zonal word $w$ such that 
			\\
			&  $\bullet$ $\sfProj(w)\in\LL(\AA^{\tscZonal})$ and $w\models\CC^{\tscZonal}$ and
         \\
         & $\bullet$ in which two consecutive positions labeled with 
         \\
         & $\qquad$symbols from $2^{\Sigma}$ have different data values?
\end{tabular}
}
\end{center}

The algorithm in Subsection~\ref{ss: algorithm diamond aut}
can be adapted to solve {\sc omega-SAT-zonal-automata}.
Extra cares are needed for the following two issues:
(1) that each two consecutive zones must be assigned different data values,
and (2) the possibility that the given zonal automaton 
accepts only $\omega$-words with finitely many zones.
We refer the reader to Appendix~\ref{app: s: proof emptiness profile diamond aut}
for the details.

\section{Two-variable logic for data $\omega$-words}
\label{s: fo2}

For the purpose of logical definability, we view data $\omega$-words as structures
\begin{equation}
\label{data-tree-eq}
w\ = \langle \nn, +1, \{a(\cdot)\}_{a\in\Sigma}, \sim \rangle,
\end{equation} 
where $\nn$ is the natural numbers $\{1,2,\ldots\}$ 
which indicates the positions, $+1$ is the successor
relation (i.e., $+1(i,j)$ iff $i+1=j$), the $a(\cdot)$'s are the
labeling predicates, and $i \sim j$ holds iff positions  $i$ and $j$
have the same data value. 

We let $\sfFO$ stand for first-order logic, 
$\sfMSO$ for monadic second-order logic 
(which extends $\sfFO$ with quantification over sets of positions),
and $\sfEMSO$ for existential monadic second order logic, i.e., 
sentences of the form $\exists X_1 \ldots \exists X_m\ \psi$, 
where $\psi$ is an $\sfFO$ formula over the vocabulary extended 
with the unary predicates $X_1,\ldots,X_m$. 
We let $\sfFO^2$ stand for $\sfFO$ with two variables,
i.e., the set of $\sfFO$ formulae that only use two variables $x$ and $y$.
The set of all sentences of the form 
$\exists X_1 \ldots \exists X_m\ \psi$, 
where $\psi$ is an $\sfFO^2$ formula is denoted by $\sfEMSO^2$. 

To emphasize that we are talking about a logic over data words we write
$(+1,\sim)$ after the logic: e.g., $\sfFO^2(+1,\sim)$ and $\sfEMSO^2(+1,\sim)$. 
Note that $\sfEMSO^2(+1)$ is equivalent in expressive power to $\sfMSO$ over the
usual (not data) finite words, i.e., 
it defines precisely the regular languages~\cite{thomas-handbook}.


It was shown in \cite{luc-lics06} that $\sfEMSO^2(+1,<,\sim)$ is
decidable over data words. In terms of complexity, the satisfiability of
this logic is shown to be at least as hard as reachability in Petri
nets. Without the $+1$ relation, the complexity drops to $\sfNEXPTIME$-complete; 
however, without $+1$ the logic is not sufficiently
expressive to capture regular relations on the data-free part of the finite word.

In this section we will prove the following:

\begin{theorem}
\label{main-thm}
The satisfiability problem is decidable for $\sfEMSO^{2}(+1,\sim)$ 
over data $\omega$-words. 
Moreover, the complexity of the decision procedure is elementary.
\end{theorem}

\subsection{A normal form for $\sfEMSO^2(+1,\sim)$} \label{sec:normal}

Decidability proofs for two-variable logics typically follow this pattern: 
a syntactic normal form is established; to be followed by
a combinatorial proof, where decidability is proved for
that normal form (by establishing the finite-model property, or by
automata techniques, for example). 

Our proof is not different that it starts by establishing a normal form for $\sfFO^2(+1,\sim)$,
and then prove the decidability for the normal form.
In fact, our normal form follows closely the one given in~\cite{luc-pods06-jacm}
for unranked finite data trees.
It can simply be adapted it to the case of $\omega$-words.
It easily follows from \cite{luc-pods06-jacm} that every $\sfEMSO^2(+1,\sim)$ sentence over data $\omega$-words is equivalent to a sentence 
$$
\exists X_1 \ldots \exists X_k (\chi \wedge \bigwedge_i \phi_i
\wedge \bigwedge_j \psi_j)
$$
where 
\begin{enumerate}
\item 
$\chi$ is an $\sfFO^2(+1)$ sentence over the extended alphabet 
 $\Sigma \times \{\ast,\top,\bot\} \times \{\ast,\top,\bot\}$
(and it can be converted to a profile B\"uchi automaton in elementary complexity);
\item 
each $\phi_i$ is of the form 
$\forall x\forall y (\alpha(x) \wedge \alpha(y) \wedge x\sim y \to x=y)$, 
where $\alpha$ is a conjunction of labeling predicates, $X_k$'s, and their negations; 
and 
\item 
each $\psi_j$ is of the form 
$\forall x\exists y\ \alpha(x) \to (x\sim y \wedge \alpha'(y))$, 
with $\alpha$, $\alpha'$ as in item 2. 
\end{enumerate}
The number of the unary predicates $X$'s is single exponential
in the size of the original input sentence.

If we extend the alphabet to $\Sigma\times 2^k$ so that each label also
specifies the family of the $X_i$'s the node belongs to, then sentences
in items 2 and 3 can be encoded by data-constraints: 
formulae in item 2 become key- and denial-constraints, 
and formulae in item 3 become inclusion-constraints.
Sentence (1) simply becomes an $\sfFO^2(+1)$ sentence over the alphabet $\Sigma \times 2^k$.

Indeed, consider, for example, the sentence $\forall x\forall y
(\alpha(x) \wedge \alpha(y) \wedge x\sim y \to x=y)$. 
Let $\Sigma'$ be the set of all symbols  
$(a,\bar b)\in\Sigma\times 2^k$ consistent with $\alpha$. 
That is, $a$ is the labeling symbol used in $\alpha$ 
(if $\alpha$ uses one) or an arbitrary letter (if $\alpha$ does not use a labeling predicate), 
and the Boolean vector $\bar b$ has $1$ in positions of the $X_i$'s 
used positively in $\alpha$ and 
$0$ in positions of $X_j$'s used negatively in $\alpha$. 
Then the original sentence is
equivalent to the key-constraints: $V(a)\mapsto a$, for each $a \in \Sigma'$
and denial-constraints: $V(a)\cap V(b)=\emptyset$, for every $a,b \in \Sigma'$ and $a\neq b$. 
The transformation of item 3 sentences into inclusion-constraints is the same. 

Hence, the satisfiability problem of $\sfEMSO^2(+1,\sim)$
can be reduced to the emptiness problem of profile B\"uchi automata
with data-constraints,
whose elementary complexity has been established in the previous section.

\section{LTL that handles data values}
\label{s: LTL data value}

In this section we extend the standard LTL
with the operators $\Diamond^{w}$, $\Diamond^{s}$, $\ttX_{\sim}$, $\ttX_{\nsim}$
to handle comparison between data values,
which we denoted by $\completeLTL$.

Let $\Sigma$ be a finite alphabet.
Formally, the logic $\completeLTL$ is defined as follows.
\begin{itemize}
\item
Both $\True$ and $\False$ are  $\completeLTL$ formulae.
\item
For each $a \in \Sigma$,
$a$ is a $\completeLTL$ formula.
\item
If $\varphi$ and $\psi$ are $\completeLTL$ formulae,
then so are
$$
\begin{array}{cccccccccccc}
\neg \varphi & ; &
\varphi \vee \psi & ; &
\varphi \wedge \psi & ; &
\ttX \;\varphi & ; &
\varphi \;\ttU\; \psi & ; &
\varphi \;\ttR\; \psi 
\end{array}
$$
\item
If $\varphi$ is a $\completeLTL$ formula,
then so are 
$$
\begin{array}{ccccccc}
\Diamond^{w} \; \varphi & ; &
\Diamond^{s} \; \varphi & ; &
\ttX_{\sim} \;\varphi & ; &
\ttX_{\nsim} \; \varphi 
\end{array}
$$
\end{itemize}
The operators $\ttX,\ttU,\ttR$
stand for neXt, Until and Release, respectively.
We write $\ttF \varphi$ as abbreviation for $\True \ttU \varphi$
and $\ttG \varphi$ for $\neg \ttF (\neg \varphi)$.
The operators $\Diamond^w \varphi$, $\Diamond^s\varphi$ are to check the existence 
of a data value in the position where the formula $\varphi$ holds.

We will not give the formal semantics of $\completeLTL$ here,
which can be found in Appendix~\ref{app: s: semantics complete LTL}.
Instead we give only the intuitive meanings of the operators 
$\Diamond^{w}$, $\Diamond^{s}$, $\ttX_{\sim}$ and $\ttX_{\nsim}$,
which are as follows.
\begin{itemize}
\item
The formula $\ttX_{\sim}$ holds in position $i$,
if it has the same data value as the next position $i+1$.
\item
The formula $\ttX_{\nsim}$ holds in position $i$,
if it has different data value as the next position $i+1$.
\item
The formula $\Diamond^{w} \varphi$ holds in position $i$,
if there exists a position $j$ that
has the same data value as position $i$ and in which the formula $\varphi$ holds.
\item
The formula $\Diamond^{s} \varphi$ holds in position $i$,
if there exists a position $j \neq i$ that
has the same data value as position $i$ and in which the formula $\varphi$ holds.
\end{itemize}
For an data $\omega$-word $w$ and a formula $\varphi \in \completeLTL$, 
we write $w,i \models \varphi$ to denote that 
in position $i$ the formula $\varphi$ holds.
As usual, for a formula $\varphi \in \completeLTL$,
we denote by $\LL(\varphi)$ the set of words $w$ for which $w,1 \models \varphi$.

Notice the subtle difference between $\Diamond^w$ and $\Diamond^s$,
with $w$ stands for ``weak'' and $s$ for ``strong,'' respectively.
With $\Diamond^w$ it is not necessary that the position $j$
is different from the current position,
while with $\Diamond^s$ the position $j$ must be different.
Obviously, $\Diamond^w$ is weaker than $\Diamond^s$,
as $\Diamond^w \varphi$ can be expressed as $\varphi \vee \Diamond^s \varphi$,
hence the name ``weak'' and ``strong.''
In fact there exists a language expressible in $\strongdiamondLTL$,
but not in $\weakdiamondLTL$.

At the first glance, it may appear that $\Diamond^w$ is too weak to capture any interesting property.
But as we will see later that the satisfiability problem even for $\weakdiamondLTL$
is already $\sfNEXPTIME$-complete.

We will denote by $\weakdiamondLTL$ and $\strongdiamondLTL$
the class of formulae that uses only $\Diamond^{w}$ and $\Diamond^{s}$, respectively,
but do not use the operators $\ttX_{\sim}$ and $\ttX_{\nsim}$.
We give some examples which will be used in the later sections.

\begin{example}
\label{eg: key}
Consider the language $L_{key(a)}$ which consists of
data words in which every two positions labeled with $a$
have different data values.
$L_{key(a)}$ is expressible by the $\strongdiamondLTL$ formula
$\ttG \: (a \to \neg \Diamond^{s} a)$.
On the other hand, the formula $\ttG \: (a \to \neg \Diamond^w a)$ 
does not make much sense as essentially 
it only expresses the data words in which the symbol $a$ does not appear.
\end{example}

\begin{example}
\label{eg: twice data values}
Consider the formula $\varphi := \ttG (a \to \Diamond^s a)$ over the alphabet $\Sigma$.
Then, $w \in \LL(\varphi)$ if and only if
every data value in $V_w(a)$ appears at least twice (among $a$-positions).
This language $\LL(\varphi)$ cannot be captured by an ADC.
\\
Now consider a slightly different representation of the formula $\varphi$.
Let $\overline{\Sigma}$ be a copy of the alphabet $\Sigma$,
in which $\overline{b} \in \overline{\Sigma}$
denotes the corresponding symbol of $b \in \Sigma$.
Consider the following formula $\varphi' := \ttG(a \to \Diamond^s \overline{a})$
over the alphabet $\Sigma \cup \overline{\Sigma}$.
Essentially $\varphi$ and $\varphi'$ are equivalent up to renaming $\overline{a}$ back to $a$.
However, $\LL(\varphi')$ can be captured by an ADC.
This simple trick will be useful in our translation of $\strongdiamondLTL$
to an ADC for the purpose of deciding the satisfiability problem for $\strongdiamondLTL$.   
\end{example}

\begin{theorem}
\label{t: complexity LTL}~
\begin{enumerate}
\item
The satisfiability problem for $\weakdiamondLTL$ is $\sfNEXPTIME$-complete.
\item
The satisfiability problem for $\strongdiamondLTL$ is 2-$\sfNEXPTIME$.
\item
The satisfiability problem for $\completeLTLstrong$ is 3-$\sfNEXPTIME$.
\end{enumerate}
\end{theorem}

The proofs for the upper bounds in Theorem~\ref{t: complexity LTL}
can be found in Appendix~\ref{app: s: proof complexity LTL}.
The proof for the hardness part in (1) 
can be found in Appendix~\ref{app: s: hardness of weak diamond LTL}.

\bibliographystyle{amsalpha}

\newpage

\appendix

\section{Analysis of the time complexity of the Algorithm in Subsection~\ref{ss: algorithm diamond aut}}
\label{app: s: complexity analysis}

Obviously Step~(1) takes exponential time in the size of the alphabet $\Sigma$.
Moreover, the sizes of the automaton $\tilde{\AA}_{\tsffin}$,
the formula $\varphi$ and the B\"uchi automaton $\tilde{\AA}$
are all exponential in the size of the original alphabet $\Sigma$.
The emptiness of B\"uchi automaton $\tilde{\AA}$ can be checked in polynomial time,
while the Presburger automaton $(\tilde{\AA}_{\tsffin},\CC)$
can be checked in $\sfNP$.
So overall our algorithm works in $\sfNEXPTIME$.

\section{Proof of the correctness of the algorithm in Subsection~\ref{ss: algorithm diamond aut}}
\label{app: s: proof correctness}

Throughout this section
we fix an ADC $(\AA,\CC)$ and
$\MM=\langle Q,\mu\rangle$ the transition system of $\AA$,
where $\AA=\MM_{q_0}^F$.
We will demonstrate the following two claims,
of which proofs are provided into the subsequent two subsections.
\begin{claim}
\label{cl: omega to finite}
Suppose there exists an data $\omega$-word $w \in \LL(\AA,\CC)$.
Then, by fixing $\SS_0 = \SS_{0}(w)$,
$\SS_{\tsffin} = \SS_{\tsffin}(w)$ and
$\SS_{\infty} = \SS_{\infty}(w)$,
the constructed Presburger automaton $(\tilde{\AA}_{\tsffin},\varphi)$
and B\"uchi automaton $\tilde{\AA}$ are both not empty.
\end{claim}

\begin{claim}
\label{cl: finite to omega}
Suppose there exist a partition $\SS_0,\SS_{\tsffin}, \SS_{\infty}$
of the set $2^{Q}-\{\emptyset\}$
such that the constructed Presburger automatn $(\tilde{\AA}_{\tsffin},\varphi)$
and B\"uchi automaton $\tilde{\AA}$ are both not empty.
Then, there exists an data $\omega$-word $w \in \LL(\AA,\CC)$ such that 
$\SS_0(w) = \SS_{0}$,
$\SS_{\tsffin}(w) = \SS_{\tsffin}$ and
$\SS_{\infty}(w) = \SS_{\infty}$.
\end{claim}

We write $w[\leq i]$ to denote the initial segment of $w$ of length $i$,
while $w[\geq i]$ the $\omega$-word obtained by discarding the initial segment of length $i-1$ from $w$.
Then, $\sfProj(w[\leq i]) = a_1\cdots a_i$, and $\sfProj(w[\geq i]) = a_i a_{i+1} \cdots$.

\subsection{Proof of Claim~\ref{cl: omega to finite}}

Let $w$ be an data $\omega$-word accepted by $(\AA,\CC)$.
Let 
$\SS_0 = \SS_0(w)$,
$\SS_{\tsffin} = \SS_{\tsffin}(w)$, and
$\SS_{\infty} = \SS_{\infty}(w)$.
Let $N$ be the minimal index $N$ such that
for each $S \in \SS_{\tsffin}$, $[S]_{w[\leq N]} = [S]_{w}$.

Let $\rho= p_1 p_2 \cdots$ be the
accepting run of $\AA$ on $\sfProj(w)$.
Let $\tilde{\Sigma}$ and $\tilde{\MM}$
be the new alphabet and the transition system
constructed in Step~(2) of our algorithm.
Then, we pick the state $p_N$ for the state $q$,
supposedly be non-deterministically picked in Step~(3) of our algorithm.
The Presburger automaton $(\tilde{\AA}_{\tsffin},\varphi)$ 
constructed in Step~(3) has the final state $p_N$,
while the B\"uchi automaton $\tilde{\AA}$
has the initial state $p_N$.
That is, $\tilde{\AA}_{\tsffin} = \tilde{\MM}_{q_0}^{\{p_N\}}$
and $\tilde{\AA}=\tilde{\MM}_{p_N}^{F}$.

Consider the (without data) $\omega$-word $x_1 x_2 \cdots$ over the alphabet $\tilde{\Sigma}$,
where
$$ 
x_i = 
\left\{
\begin{array}{ll}
a_i \in \Sigma & \mbox{if} \ d_i \in [S]_{w} \ \mbox{and} \ S \notin \SS_{\infty}
\\
(a_i,S) \in \Sigma_{S}& \mbox{if} \ d_i \in [S]_{w} \ \mbox{and} \ S \in \SS_{\infty}
\end{array}
\right.
$$
We claim that the following words:
\begin{itemize}
\item
$v_1 = x_1 x_2 \cdots x_N \in \LL(\tilde{\AA}_{\tsffin},\varphi)$.
\item
$v_2 = x_{N+1} x_{N+2} \cdots \in \LL(\tilde{\AA})$.
\end{itemize}

\subsubsection{Proof of $v_1 = x_1 x_2 \cdots x_N \in \LL(\tilde{\AA}_{\tsffin},\varphi)$}

There are two things to show here:
\begin{enumerate}
\item
That $v_1$ is accepted by $\tilde{\AA}_{\tsffin}$.
\item
That $\varphi(\sfParikh(x_1\cdots x_N))$ holds.
\end{enumerate}
It is pretty straightforward to verify that $p_1\cdots p_N$
is a run of $\AA$ on $x_1\cdots x_N$.
That it is an accepting run follows from the fact that
$q_N$ is a final state in $\tilde{\AA}_{\tsffin}$.

Now we will show that $\varphi(\sfParikh(q_1\cdots q_N))$ holds.
Recall that the formula $\varphi$ is of the form:
$$
\exists z_{S_1} \; \cdots \; \exists z_{S_m} \ 
\psi_1 \wedge
\psi_2 \wedge
\psi_3 \wedge
\psi_4 
$$
where 
\begin{itemize}
\item
the formula $\psi_1$ is the conjunction
$$
\bigwedge_{a \in \Sigma} x_{a} \geq \sum_{S \ni a} z_S
$$
\item
the formula $\psi_2$ is the conjunction
$$
\bigwedge_{S \in \sSS_0 \cup \sSS_{\infty}} z_S = 0 
$$
\item
the formula $\psi_3$ is the conjunction
$$
\bigwedge_{S \in \sSS_{\tsffin}} z_S \geq 1
$$
\item
the formula $\psi_4$ is the conjunction
$$
\bigwedge_{V(a)\mapsto a \in \sCC} 
x_{a} = \sum_{a \in S} z_S
$$
\end{itemize}
In order to show that $\varphi(\sfParikh(\sfProj(w)))$ holds,
for each $S \subseteq Q$, 
we pick the following integers as witnesses for $z_S$.
\begin{itemize}
\item
$z_{S} = |[S]_{w[\leq N]}|$, for each $S \in \SS_{\tsffin}(w)$.
\item
$z_S = 0$, for each $S \notin \SS_{\tsffin}(w)$.
\end{itemize}
We need to show that all the formulae $\psi_1$--$\psi_4$ above are satisfied.  

First, we observe that the following two points.
For each $a \in \Sigma$,
\begin{enumerate}
\item
$\#_a(v_1)$ is precisely the number of $a$-positions in $w[\leq N]$
whose data value is from the set 
$$
\bigcup_{S \in \sSS_{\tsffin}} [S]_w
$$
\item
$\#_(a,S)(v_1)$ is precisely the number of $a$-positions in $w[\leq N]$
whose data value is from the set $[S]_w$.
Recall that in this case $S \in \SS_{\infty}$.
\end{enumerate}
Then, $\psi_1$ follows immediately from~(1)
that such $\#_a(v_1)$ number of $a$-positions must be greater than 
the number of its data values $\sum_{S \in \sSS_{\tsffin}} |[S]_w|$.
The formulae $\psi_2$ and $\psi_3$ follows immediately from the definition.
That the formula $\psi_4$ holds is because of~(1)
and that the number $\#_a(v_1)$ of such $a$-positions
is precisely the number of its data value $\sum_{S \in \sSS_{\tsffin}} |[S]_w|$.

\subsubsection{Proof of $v_2 = x_{N+1} x_{N+2} \cdots \in \LL(\tilde{\AA})$}

Recall that the B\"uchi automaton $\tilde{A}$ is
the intersection of $\tilde{\MM}_{p_N}^{F}$ with 
the automaton that checks the following condition.
\begin{enumerate}
\item
Each $(a,S) \in \Sigma(\SS_{\infty})$ appears infinitely many times.
\item
If the key-constraint $V(a) \mapsto a \in \CC$, 
then the symbol $a$ does not appear.
\end{enumerate}
Now, to show that $v_2 \in \LL(\tilde{\AA})$,
we claim that $p_{N+1} p_{N+2} \cdots$
is also an accepting run of $\tilde{\AA}$ on $v_2$.

First, we show that $x_{N+1}x_{N+2}\cdots$
satisfies the properties (1) and (2) above.
As $S \in \SS_{\infty}(w)$,
then it means each data values in $[S]_w$ appears infinitely many often in $w$.
By our construction of $x_{N+1} x_{N+2}\cdots$,
it means each symbol $(a,S) \in \Sigma(\SS_{\infty})$ 
appears infinitely many often in $x_{N+1} x_{N+2} \cdots$.

Furthermore, recall that $N$ is an index such that
$[S]_{w[\leq N]} = [S]_w$, for each $S \in \SS_{\tsffin}(w)$.
Now, if $w \models V(a)\mapsto a$,
then every $a$-position greater than $N$ in $w$
has data value from the set $\bigcup_{S \in \sSS_{\infty}(w)} [S]_w$.
This means that by our construction of $x_{N+1} x_{N+2}\cdots$,
the symbol $a$ does not appear in $x_{N+1} x_{N+2} \cdots$.

To show that $x_{N+1} x_{N+2} \cdots$ is accepted by $\tilde{\AA} = \tilde{\MM}_{p_N}^{F}$,
we observe that $p_{N+1} p_{N+2} \cdots$ 
is an accepting run of $\tilde{\AA}$ on $x_{N+1}x_{N+2}\cdots$,
which is immediate by our construction of $\tilde{\MM}$.

\subsection{Proof of Claim~\ref{cl: finite to omega}}

Suppose there are the following items:
\begin{itemize}
\item
$\SS_0$, $\SS_{\tsffin}$, $\SS_{\infty}$ is a partition of $2^{\Sigma}-\{\emptyset\}$;
\item
$\tilde{\Sigma} = \Sigma \cup \Sigma(\SS_{\infty})$ and $\tilde{\MM} = \langle \tilde{Q}, \tilde{\mu}\rangle$
be the constructed new alphabet and transition system;
\item
a state $q\in \tilde{Q}$,
\end{itemize}
such that the constructed the Presbruger automaton $(\tilde{\AA}_{\tsffin},\varphi)$
and the B\"uchi automaton $\tilde{\AA}$ are not empty.
Consider the following two words.
\begin{itemize}
\item
$v_1 = b_1 \cdots b_N \in \LL(\tilde{\AA}_{\tsffin},\varphi)$,
where $p_1\cdots p_N$ be an accepting run of $\tilde{\AA}_{\tsffin}$ on $v_1$.
\item
$v_2 = b_{N+1} b_{N+2} \cdots  \in \LL(\tilde{\AA})$,
where $p_{N+1} p_{N+2}\cdots$ be an accepting run of $\tilde{\AA}$ on $v_2$.
\end{itemize}
We will construct an data $\omega$-word $w \in (\Sigma \times \fD)^{\omega}$
$$
w = {a_1 \choose d_1}{a_2 \choose d_2} \cdots {a_N \choose d_N}{a_{N+1} \choose d_{N+1}} \cdots,
$$
which is accepted by $(\AA,\CC)$.

We start by defining $\sfProj(w) = a_1 a_2 \cdots$.
For each $i=1,2,\ldots$,
$$
a_i =
\left\{
\begin{array}{ll}
b_i & \mbox{if} \ b_i \in \Sigma
\\
c  & \mbox{if} \ b_i = (c,S) \in \Sigma(\SS_{\infty}) \ \mbox{for some} \ S
\end{array}
\right.
$$
By the construction of $\tilde{\AA}_{\tsffin}$ and $\tilde{\AA}$,
it is immediate that $p_1 p_2 \cdots p_{N+1} p_{N+2}\cdots$
is an accepting run of $\AA$ on $\sfProj(w)$.

Now we will define the data values $d_1,d_2,\ldots$.
For each $S \in \SS_{\infty}$,
we fix an infinite set of data values for $\fD_{S}$,
such that all those sets $\fD_S$'s are disjoint.
We will use $\fD_S$ for $[S]_{w}$ for each $S \in \SS_{\infty}$.

For each $S \in \SS_{\tsffin}$,
the set $[S]_{\rho(w)}$ can be computed as follows.
By the assumption, $v_1$ is a word such that $\varphi(\sfParikh(v_1))$ holds,
where $m_S$ is a witness for the variable $z_S$.
Let $K = \sum_{S} m_S$.
Define a function
$$
\xi : \{1,\ldots,K\} \to 2^{\Sigma}-\{\emptyset\},
$$
such that $|\xi^{-1}(S)| = m_S$.
We will use $\xi^{-1}(S)$ as $[S]_{\rho(w)} = \rho_{\rho(w)[\leq N]}$.

The assignment of data values to $w$ can be done as follows.
\begin{enumerate}
\item 
We first define the data values for $d_1, \ldots, d_N$.
For each $a \in \Sigma$, pick the positions $Z(a) = \{i \mid b_i = a\}$.
(Note that the parameter in defining the set $Z(a)$ of positions
is the word $b_1 \cdots b_N$.)
Then we can assign those positions in $Z(a)$ 
with the data values from $\bigcup_{a \in S} \xi^{-1}(S)$.
Such assignment is possible as $|Z(a)| = \#_{v_1}(a) \geq \sum_{a \in S} m_S$.
\item
Then, we define the data values $d_{N+1}, d_{N+2}, \ldots$,
where $b_i\in \Sigma$.
This is easy. We just pick some arbitrary data values from $\bigcup_{q \in S} \xi^{-1}(S)$.
\item
At this stage we have define all the data values $d_i$'s
for the positions $i$ labeled with symbols from $\Sigma$ in the $v_1v_2$.
What is left is to define the data values for the positions in $v_1 v_2$
whose labels are from $\Sigma(\SS_{\infty})$.
Here we will use the data values in $\fD_S$ and the assignment is done inductively.
For each data value $d$ in $\fD_S$ that has not appeared yet in $w$,
we pick $|S|$ number of positions $l_1,\ldots,l_{|S|}$ in $w$ such that
$\{a_{l_1},\ldots,a_{l_{|S|}}\} = S$ and have no data values yet.
Then, we assign all those positions with the data value $d$.
By the acceptance criteria of the B\"uchi automaton $\tilde{\AA}$,
there are infinitely many such positions for each $S \in \SS_{\infty}$.
Thus, such assignment is always possible.
\end{enumerate}
What remains now is to prove that $w \models \CC$.

By Proposition~\ref{p: data-constraint} and the construction of $\SS_0$,
as well as the Presburger formula $\varphi$,
it is immediate that $w$ satisfies the inclusion- and denial-constraints in $\CC$.
We will show that it also satisfies the key-constraints.

Suppose the key-constraint $V(a)\mapsto a \in \CC$.
First, in the assignment of data values in $w[\leq N]$
all the $a$-positions recieve different data values,
due to the constraint $|Z(a)| = \#_{v_1}(a) = \sum_{a \in S} m_S$.
Second, from the construction of the automaton $\tilde{\AA}$,
the symbol $a$ does not appear in $b_{N+1} b_{N+2}\cdots$,
thus not appearing in $w[\geq N+1]$.
This means that we do not assign any data values from $\bigcup_{a \in S} \xi^{-1}(S)$
in every $a$-positions $\geq N+1$,
so all data values in $\bigcup_{a \in S} \xi^{-1}(S)$ appears only in once in $a$-positions.
Lastly, all the data values in $[S]_{\rho(w)}$ for each $S \in \SS_{\infty}$
are assigned only once.
Thus, it follows that every $a$-positions in $w$
have different data values, thus, $w \models V(a) \mapsto a$.
This completes the proof of Claim~\ref{cl: finite to omega}.

\section{The $\sfNP$ algorithm for Theorem~\ref{t: emptiness diamond aut}}
\label{app: s: proof np emptiness diamond aut}

We identify that in our algorithm in Subsection~\ref{ss: algorithm diamond aut},
the exponential blow-up occurs in Step~(1),
where we have to enumerate all the non-empty subsets of $\Sigma$.
Especially, the size of the set $\SS_{\infty}$
determines the sizes of the new alphabet $\tilde{\Sigma}$,
the transition system $\tilde{\MM}$.
And the size of the set $\SS_{\tsffin}$
determines the size of the Presburger formula $\varphi$.

The main idea of our $\sfNP$ is that
if there is no key-constraint in $\CC$, then
the following holds.
There exists a subset $\ZZ \subseteq 2^{\Sigma}$ of polynomial size
such that there exists an data $\omega$-word $w \in \LL(\AA,\CC)$
if and only if
there exists an data $\omega$-word $w' \in \LL(\AA,\CC)$,
where $[S]_{w'} = \emptyset$, for all $S \notin \ZZ$.
This means that in the constructions of
$\tilde{\Sigma}$, $\tilde{\MM}$, and $\varphi$,
we only need to take into account the sets in $\ZZ$.
This idea is the one that we are going to explain in the next subsection.

\subsection{Preliminary notion}
\label{app: ss: notion for np algorithm}

Let $\CC$ be a collection of inclusion- and denial-constraints.
We define the subset $\SS_0(\CC) \subseteq 2^{\Sigma}$ as follows. 
\begin{enumerate}
\item
If $\CC$ contains the inclusion-constraint $V(a) \subseteq \bigcup_{b\in R} V(b)$,
then $S \in \SS_0(\CC)$ for all $S \subseteq \Sigma$ where $a\in S$ and $S \cap R = \emptyset$.
\item
If $\CC$ contains the denial-constraint $V(a) \cap V(b) = \emptyset$,
then $S \in \SS_0(\CC)$ for all $S \subseteq \Sigma$ where $S$ contains both $a$ and $b$.
\end{enumerate}

\begin{remark}
\label{r: identify the empty-set}
Given a non-empty set $S \subseteq \Sigma$,
we can decide in polynomial time whether $S \in \SS_0(\CC)$
\end{remark}

\subsection{The algorithm}
\label{app: ss: np algorithm}

Given an ADC $(\AA,\CC)$, where $\CC$ does not contain key-constraints,
the algorithm works as follows.
\begin{enumerate}
\item
Construct (non-deterministically) a function $f: \Sigma \mapsto 2^{\Sigma}$ such that
for each $a\in \Sigma$, either
$$
a \in f(a) \ \mbox{and} \ f(a) \notin \SS_0(C)
$$
or
$$
f(a) = \emptyset
$$
Such function can be non-deterministically constructed,
by guessing $f(a)$ for each $a \in \Sigma$
and verify (in polynomial time) deterministically that $f(a)\notin \SS_0(\CC)$.
\item
Divide (non-deterministically) $\sfImage(f)$ into two categories: 
$\SS_{\tsffin}$ and $\SS_{\infty}$.
\\
The intended meaning of $\SS_{\tsffin}$ and $\SS_{\infty}$
is the same as the algorithm in Subsection~\ref{ss: algorithm diamond aut}.
Every other subsets not in $\SS_{\tsffin} \cup \SS_{\infty}$
are considered in $\SS_0$.
\item
Define the alphabet $\tilde{\Sigma}=\Sigma \cup \Sigma(\SS_{\infty})$,
where 
$$
\Sigma(\SS_{\infty}) = \{(a,S) \mid a \in S \ \mbox{and} \ S \in \SS_{\infty}\},
$$ 
and 
a transition system $\tilde{\MM} =\langle \tilde{Q}, \tilde{\mu}\rangle$
over the alphabet $\tilde{\Sigma}$, as follows.
\begin{eqnarray*}
\tilde{Q} & = & Q 
\\
\tilde{\mu} & = & \mu \cup
\{(p,(a,S),q) \mid (p,a,q) \in \mu \ \mbox{and} \ (a,S)\in \Sigma(\SS_{\infty})\}
\end{eqnarray*}
\item
Non-deterministically choose one state $q \in \tilde{Q}$.
\item
Construct a Presburger automaton $(\tilde{\AA}_{\tsffin},\varphi)$, where 
$\tilde{\AA}_{\tsffin} = \tilde{\MM}_{q_0}^{\{q\}}$ and
the formula $\varphi$ is as follows.
\\
Let $\sfImage(f) = \{S_1,\ldots,S_l\}$.
Then $\varphi$ is of the form 
$\exists z_{S_1} \; \cdots \; \exists z_{S_l}\ \psi$, 
where $\psi$ is the following quantifier-free formula:
$$
\bigwedge_{a \in \Sigma}
\Bigg(
x_{a} \geq \sum_{S \ni a \ \mbox{\tiny and} S \in \tsfImage(f)} z_S
\Bigg)
\quad
\wedge
\quad
\bigwedge_{S \in  \sSS_{\infty}} z_S = 0 
\quad
\wedge
\quad
\bigwedge_{S \in \sSS_{\tsffin}} z_S \geq 1
$$

\item
Construct a B\"uchi automaton $\tilde{\AA}$is
simply the intersection of $\tilde{\MM}_{q}^{F}$
with the automaton that checks that each $(a,S) \in \Sigma(\SS_{\infty})$
appears infinitely many times.

\item
Test the emptiness of $\LL(\tilde{\AA}_{\tsffin},\varphi)$ and $\LL(\tilde{\AA})$.
\\
Then, $\LL(\AA,\CC)\neq \emptyset$ if and only if $\LL(\tilde{\AA}_{\tsffin},\varphi) \neq \emptyset$
and $\LL(\tilde{\AA})\neq \emptyset$.
\end{enumerate}

\subsection{The proof of correctness}
\label{app: ss: correctness np algorithm}

In view of Claims~\ref{cl: omega to finite} and~\ref{cl: finite to omega},
to prove the correctness of our algorithm,
it is sufficient to prove the following.
\begin{claim}
\label{cl: many empty-sets}
If an data $\omega$-word $w \models \CC$,
then there exist a function $f:\Sigma \mapsto 2^{\Sigma}$
that respects the condition in Step~(1) 
and an data $\omega$-word $v \models \CC$
such that $\sfProj(v)=\sfProj(w)$ and $[S]_v=\emptyset$,
for all $S \notin \sfImage(f)$.
\end{claim}
\begin{proof}
Let
$$
w = {a_1 \choose d_1} {a_2 \choose d_2} \cdots.
$$
We define the function $f$ as follows.
For each $a\in \Sigma$,
\begin{itemize}
\item
if the label $a$ does not appear in $w$,
then $f(a) = \emptyset$;
\item
otherwise, define $f(a) = S_a$ such that $a\in S_a$
and $[S_a]_w \neq \emptyset$.
\\
Such a set $S_a$ exists as there is at least one $a$-position in $w$
and this position has a data value in $V_w(a)$,
which is partitioned into $\bigcup_{a \in S} [S]_w$.
\end{itemize}
We define the data word $v$ as follows.
$$
v = {a_1 \choose d_1'} {a_2 \choose d_2'} \cdots.
$$
Thus, $\sfProj(v) = \sfProj(w)$.
We define the data values $d_1',d_2',\ldots$ as follows.
\begin{itemize}
\item
If $d_i \in [S]_w$, for some $S \in \sfImage(f)$,
then $d_i' = d_i$.
\item
If $d_i \notin [S]_w$, for all nonempty $S \in \sfImage(f)$,
then we pick arbitrary data value from $[f(a_i)]_w$
to assign to $d_i'$.
\end{itemize}
By such construction, we have $[S]_v = [S]_w$,
for all non-empty $S \in \sfImage(f)$.
By Proposition~\ref{p: data-constraint},
$v \models \CC$.
Furthermore, $[S]_v = \emptyset$, for all $S \notin \sfImage(f)$.
This completes the proof of our claim.
\end{proof}

\section{Proof of Theorem~\ref{t: emptiness profile diamond aut}}
\label{app: s: proof emptiness profile diamond aut}

For the sake of presentation,
we first show the decidability of a simpler version
of the problem {\sc Omega-SAT-zonal-automata},
which we call {\sc Omega-SAT-locally-different} in Subsection~\ref{app: ss: locally different}.
Then, in Subsection~\ref{app: ss: zonal automata}
we explain how to adapt the approach in Subsection~\ref{app: ss: locally different}
for {\sc Omega-SAT-zonal-automata}.

\subsection{Locally different data $\omega$-words}
\label{app: ss: locally different}

A data word $w = {a_1 \choose d_1}{a_2 \choose d_2}\cdots$
is called {\em locally different},
if each position has different data value from 
its left- and right-neighbors,
that is, $d_{i} \neq d_{i+1}$, for each $i=1,2,\ldots$.

In this section we give an algorithm to decide the problem
{\sc SAT-locally-different} defined below. 
\begin{center}
\fbox{
\begin{tabular}{ll}
{\sc Problem}: & {\sc Omega-SAT-locally-different} 
\\ \hline
{\sc Input}: & a B\"uchi automaton $\AA$ and  
\\ 
             & a collection $\CC$ of key-, inclusion- and denial-constraints
\\
{\sc Question}: & is there a locally different data word $w$ such that 
\\
               & $\sfProj(w)\in\LL(\AA)$ with an accepting run $\rho$ and $\rho(w)\models\CC$?
\end{tabular}
}
\end{center}

In the proof we will use the following simple lemma.

\begin{lemma}
\label{l: swapping data}
{\em \cite[Lemma 3]{fo2-lpar}}
Let $v$ be a finite data word over $\Sigma$.
Suppose that for each $a \in \Sigma$,
either $V_v(a) = \emptyset$ or $|V_v(a)| \geq |\Sigma|+3$.
Then we can rearrange the positions of the data values in $v$ such that the
resulting data word $w$ is locally different, $\sfProj(w) = \sfProj(v)$ and 
for each $a \in \Sigma$, $V_{w}(a) = V_{v}(a)$. 
\end{lemma}

What this lemma tells us is that when 
the number of data values in found in $a$-positions is big enough,
for each $a \in \Sigma$, then
to solve {\sc SAT-locally-different},
it is sufficient to solve {\sc Omega-SAT-ADC}.
Then, Lemma~\ref{l: swapping data} allows us to rearrange 
the data values in the solution of {\sc Omega-SAT-ADC} to be locally different.

In the rest of this section, the symbol $\varepsilon$ denotes the constant $|\Sigma|+3$.
The main idea follows roughly as the one in the previous section,
with the notable exception that for an data $\omega$-word $w$,
we divide the non-empty subsets $S \subseteq \Sigma$ into four categories:
\begin{itemize}
\item
$\SS_{0}(w) = \{S \mid [S]_{w} = \emptyset\}$.
\item
$\SS_{\tsffin}^{< \varepsilon}(w) = 
\{S \mid [S]_{w} \ \mbox{is a finite set of cardinality} \ < \varepsilon  \}$.
\item
$\SS_{\tsffin}^{\geq \varepsilon}(w) = 
\{S \mid [S]_{w} \ \mbox{is a finite set of cardinality} \ \geq \varepsilon  \}$.
\item
$\SS_{\infty}(w) = \{S \mid [S]_{w} \ \mbox{is an infinite set}  \}$.
\end{itemize}
Note that in an data $\omega$-word $w$,
for $a \in S$ and $S \in \SS_{\tsffin}^{\geq \varepsilon}(w)$,
then $V_w(a)\geq \varepsilon$.
This will allow us to apply Lemma~\ref{l: swapping data},
for $V_w(a)$, where $a \in S$ and $S \in \SS_{\infty}^{\geq \varepsilon}(w)$.
On the other hand, the data values in the sets $[S]_w$,
where $S \in \SS_{\tsffin}^{< \varepsilon}(w)$
can be regarded as fixed constants,
thus, can be embedded as part of the input alphabet.
This is our main idea to solve {\sc SAT-locally-different}.

The details are as follows.
Given an input $(\AA,\CC)$,
our algorithm does the following.
\begin{enumerate}
\item
Guess a partition $\SS_0,\SS_{\tsffin}^{< \varepsilon}, 
\SS_{\tsffin}^{\geq \varepsilon},\SS_{\infty}$ of the sets $2^Q - \{\emptyset\}$
as in the algorithm in Subsection~\ref{ss: algorithm diamond aut}.
\\
That is, it respects the following conditions.
\begin{itemize}
\item[C1.]
If the inclusion-constraint $V(a) \subseteq \bigcup_{b\in R} V(b)$ is in $\CC$,
then all the sets $S$, where $a\in S$ and $S \cap R =\emptyset$, are in $\SS_0$.
\item[C2.]
If the denial-constraint $V(a)\cap V(b) = \emptyset$ is in $\CC$,
then all the sets $S$, which contains both $a$ and $b$, are in $\SS_0$.
\end{itemize}

\item
Then, for each $S \in \SS_{\tsffin}^{< \varepsilon}$,
we further guess a non-zero constant $K_S < \varepsilon$
and fix a set $\Gamma_S$ of $K_S$ number of constants.
Define
$$
\Sigma(\SS_{\tsffin}^{< \varepsilon})
= 
\{ (a,d) \mid a \in S \ \mbox{and} \ d \in \Gamma_S \ \mbox{where} 
\ S \in \SS_{\tsffin}^{< \varepsilon}\}.
$$
The intention is that we only need to consider the data $\omega$-words $w$
in which $[S]_{w} = \Gamma_S$, for each $S \in \SS_{\tsffin}^{<\varepsilon}$.
\item
Let $\Sigma(\SS_{\infty}) = \{(a,S) \mid a \in S \ \mbox{and} \ S \in \SS_{\infty}\}$.
Construct the new alphabet $\hat{\Sigma}$, where
$$
\hat{\Sigma} = \Sigma \:\cup\: \Sigma(\SS_{\tsffin}^{< \varepsilon}) \:\cup\: \Sigma(\SS_{\infty}),
$$
and the new transition system $\tilde{\MM} = \langle \tilde{Q}, \tilde{\mu}\rangle$ 
is defined as:
\begin{eqnarray*}
\hat{Q} & = & Q 
\\
\hat{\mu} & = & \mu \cup
\left\{
\begin{array}{c}
\{(p,(a,S),q) \mid (p,a,q) \in \mu \ \mbox{and} \ (a,S) \in \SS(\infty)\}
\\
\cup
\\
\{(p,(a,d),q) \mid (p,a,q) \in \mu \ \mbox{and} \ (a,d) \in \Sigma(\SS_{\tsffin}^{< \varepsilon}) \} 
\end{array}
\right.
\end{eqnarray*}

\item
Non-deterministically choose one state $q \in \tilde{Q}$.

\item
Construct a Presburger automaton $(\tilde{\AA}_{\tsffin},\varphi)$
as follows.
\begin{enumerate}
\item
The automaton $\tilde{\AA_{\tsffin}}$ is $\MM_{q_0}^{\{q\}}$
intersect with an automaton that checks the property:
\begin{itemize}
\item
If two symbols $(a,d_1),(b,d_2) \in \Sigma(\SS_{\tsffin}^{< \varepsilon})$
appear in two consecutive positions, then $d_1 \neq d_2$.
\item
If the key-constraint $V(a)\mapsto a \in \CC$,
then the symbol $(a,d) \in \Sigma(\SS_{\tsffin}^{< \varepsilon})$.
\end{itemize}
\item
The Presburger formula $\varphi$ is defined as follows.
Let $S_1,\ldots,S_{m}$ be the enumeration of 
non-empty subsets of $Q$, where $m = 2^{|Q|}-1$.  

The formula $\varphi$ is of the form 
$\exists z_{S_1} \; \cdots \; \exists z_{S_m}\ \psi$, 
where $\psi$ is the following quantifier-free formula:
\begin{eqnarray*}
 & & \qquad\bigwedge_{a \in \Sigma} \qquad
\Big(x_{a} \geq \sum_{S \ni a} z_S\Big)
\\
 & \wedge &
\bigwedge_{S \in \sSS_0 \cup \sSS_{\tsffin}^{< \varepsilon} \cup \sSS_{\infty}} z_S = 0 
\quad \wedge \quad
\bigwedge_{S \in \sSS_{\tsffin}^{\geq \varepsilon}} z_S \geq 1
\\
& \wedge &
\quad\bigwedge_{V(a)\mapsto a \in \sCC} \quad
\Big( x_{a} = \sum_{a \in S} z_S \Big)
\end{eqnarray*}
\end{enumerate}
\item
Construct a B\"uchi automaton $\tilde{\AA}$ as follows.
\\
The B\"uchi automaton $\tilde{\AA}$ is
simply the intersection of $\tilde{\MM}_{q}^{F}$
with the automaton that checks the following condition.
\begin{enumerate}
\item
If two symbols $(a,d_1),(b,d_2) \in \Sigma(\SS_{\tsffin}^{< \varepsilon})$
appear in two consecutive positions, then $d_1 \neq d_2$.
\item
Each $(a,S) \in \Sigma(\SS_{\infty})$
appears infinitely many times.
\item
If the key-constraint $V(a) \mapsto a \in \CC$,
then the symbols $a$ and $(a,d) \in \Sigma \times \Gamma_S$,
for some $S \in \SS_{\tsffin}^{< \varepsilon}$ do not appear. 
\end{enumerate}

\item
Test the emptiness of $\LL(\tilde{\AA}_{\tsffin},\varphi)$ and $\LL(\tilde{\AA})$.
\\
Then, $\LL(\AA,\CC)\neq \emptyset$ if and only if $\LL(\tilde{\AA}_{\tsffin},\varphi) \neq \emptyset$
and $\LL(\AA)\neq \emptyset$.
\end{enumerate}
The sizes of the automaton $\tilde{\AA}_{\tsffin}$, 
the formula $\varphi$ and the B\"uchi automaton $\tilde{\AA}$
are all exponential in the size of $(\AA,\CC)$,
thus, establishing the $\sfNEXPTIME$ upper bound for {\sc SAT-locally-different}.
The proof of correctness is similar to the proofs of the Claims~\ref{cl: finite to omega}
and~\ref{cl: omega to finite}.
Lemma~\ref{l: swapping data} ensures us that
we get a locally different data $\omega$-words.
The constant data values from $\Sigma(\SS_{\tsffin}^{< \varepsilon})$
are already ensured by the automata $\tilde{\AA}_{\tsffin}$ and $\tilde{\AA}$ that 
each of them does not appear in two consecutive positions.

\subsection{The algorithm for {\sc Omega-SAT-zonal-automata}}
\label{app: ss: zonal automata}

Now we explain how the algorithm for {\sc Omega-SAT-locally-different}
can be adapted for {\sc Omega-SAT-zonal-automata}.
It works as follows.

Given a zonal automaton $\AA$
and a collection $\CC$ of data-constraints over the alphabet $2^{\Sigma}$,
the algorithm does the following.
It guesses if there exists a zonal word with infinitely many zones.
If there is one, then  
the algorithm for {\sc Omega-SAT-locally-different}
can be adapted in a straightforward manner.
Otherwise, it does the following.
Let $\MM=\langle Q, \mu \rangle$ be the transition system of $\AA$,
where $\AA = \MM_{q_0}^F$.
\begin{enumerate}
\item
Guess a state $q \in Q$.
\item
The presburger automaton $(\tilde{\AA}_{\tsffin},\varphi)$
is $\tilde{\AA}_{\tsffin} = \MM_{q_0}^{\{q\}}$ over the alphabet $\Sigma \cup 2^{\Sigma}$ and
the formula $\varphi$ can be constructed like 
in Step~(5) of the algorithm in Subsection~\ref{app: ss: locally different},
but over the alphabet $2^{\Sigma}$.
\item
The B\"uchi automaton $\tilde{\AA}$ is simply $\tilde{\MM}_{q}^F$
intersects with an automaton that checks that
the symbols from $2^{\Sigma}$ does not appear.
\end{enumerate}
The intuition is that
since there are only finitely many number zones,
all the zones and its data-constraints
are taken care by the Presburger automaton $(\tilde{\AA}_{\tsffin},\varphi)$.
The B\"uchi automaton $\tilde{\AA}$
simply makes sure that the last zone
has the property desired by the original B\"uchi automaton $\AA$.

\section{The formal semantics of $\completeLTL$}
\label{app: s: semantics complete LTL}

Formally the semantics of $\completeLTL$ is given as follows.
Let $w = {a_1 \choose d_1} {a_2 \choose d_2}\cdots$
and $i \in \{1,2,\ldots\}$.
\begin{itemize}
\item
$w,i \models \True$ and $w,i \not\models \False$;
\item
$w,i \models a$ if and only if $a_i = a$;
\item
$w,i \models \varphi \vee \psi$ if and only if 
$w,i \models \varphi$ or $w,i \models \psi$;
\item
$w,i \models \neg \varphi$ if and only if 
$w,i \models \varphi$ is not true;
\item
$w,i \models \ttX \varphi$ if and only if 
$i+1 \leq n$ and $w,i+1 \models \varphi$;  
\item
$w,i \models \ttX_{\sim} \varphi$ if and only if 
$i+1 \leq n$ and $d_{i} = d_{i+1}$ and $w,i+1 \models \varphi$;  
\item
$w,i \models \ttX_{\nsim} \varphi$ if and only if 
$i+1 \leq n$ and $d_{i} \neq d_{i+1}$ and $w,i+1 \models \varphi$;  
\item
$w,i \models \varphi \ttU \psi$ if and only if 
there exists $j \geq i$ such that
for all $i' = i,\ldots,j-1$,
$w,i' \models \varphi $ and $w,j \models \psi$; 
\item
$w,i \models \varphi \ttR \psi$ if and only if 
if there exists $j \geq i$ such that
$w,j \not\models \psi$, then
there exists $i' \in \{ i,\ldots,j-1 \}$,
$w,i' \models \varphi $; 
\item
$w,i \models \Diamond^w \varphi$ if and only if
there exists $j$ such that
$d_j = d_i$ and $w,j \models \varphi$;
\item
$w,i \models \Diamond^s \varphi$ if and only if
there exists $j\neq i$ such that
$d_j = d_i$ and $w,j \models \varphi$.
\end{itemize}
A $\completeLTL$ formula $\varphi$ defines a data language via
$L(\varphi) = \{w \mid w,1 \models \varphi\}$.

\section{Proofs of the upper bounds in Theorem~\ref{t: complexity LTL}}
\label{app: s: proof complexity LTL}

We first establish a normal form for formula in $\completeLTL$.
A formula $\varphi$ is in {\em normal form},
if every subformula in $\varphi$ that starts with a negation, say $\neg \psi$,
then $\psi$ is either $a \in \Sigma$, or
$\Diamond^s \psi'$, or $\Diamond^w \psi'$, for some $\psi'$.

\begin{proposition}
\label{p: normal LTL}
Every formula $\varphi$ in $\completeLTL$
can be converted to its equivalent normal form $\widetilde{\varphi}$ in linear time.
\end{proposition}
\begin{proof}
The construction $\widetilde{\varphi}$ is done inductively.
\begin{itemize}
\item 
If $\varphi$ does not start with a negation,
then $\widetilde{\varphi}$ is precisely $\varphi$.
\item 
If $\varphi$ is in the form $\neg \neg \psi$,
then $\widetilde{\varphi}$ is $\widetilde{\psi}$.
\item 
If $\varphi$ is in the form $\neg (\psi \vee \psi')$,
then $\widetilde{\varphi}$ is $\widetilde{\neg\psi} \wedge \widetilde{\neg\psi'}$.
\item 
If $\varphi$ is in the form $\neg (\psi \wedge \psi')$,
then $\widetilde{\varphi}$ is $\widetilde{\neg\psi} \vee \widetilde{\neg\psi'}$.
\item 
If $\varphi$ is in the form $\neg \ttX \psi$,
then $\widetilde{\varphi}$ is $\ttX \widetilde{\neg\psi}$.
\item 
If $\varphi$ is in the form $\neg (\psi \ttU \psi')$,
then $\widetilde{\varphi}$ is $\widetilde{\neg\psi} \ttR \widetilde{\neg\psi'}$.
\item 
If $\varphi$ is in the form $\neg (\psi \ttR \psi')$,
then $\widetilde{\varphi}$ is $\widetilde{\neg\psi} \ttU \widetilde{\neg\psi'}$.
\item
If $\varphi$ is in the form $\neg (\ttX_{\sim} \psi)$,
then $\widetilde{\varphi}$ is $(\ttX_{\nsim} \True) \vee \ttX_{\sim} \widetilde{\neg\psi}$.
\item
If $\varphi$ is in the form $\neg (\ttX_{\nsim} \psi)$,
then $\widetilde{\varphi}$ is $(\ttX_{\sim} \True) \vee \ttX_{\nsim} \widetilde{\neg\psi}$.
\item
If $\varphi$ is in the form $\neg \Diamond^{w} \psi$,
then $\widetilde{\varphi}$ is $\neg \Diamond^{w} \widetilde{\psi}$.
\item
If $\varphi$ is in the form $\neg \Diamond^{s} \psi$,
then $\widetilde{\varphi}$ is $\neg \Diamond^{s} \widetilde{\psi}$.
\end{itemize}
That $\varphi$ and $\widetilde{\varphi}$ are equivalent
is straightforward.
\end{proof}

\begin{remark}
It is straightforward from the construction of $\widetilde{\varphi}$,
that $\widetilde{\varphi}$ stay in the same class as $\varphi$.
That is,
\begin{itemize}
\item 
if $\varphi \in \weakdiamondLTL$, then $\widetilde{\varphi} \in \weakdiamondLTL$;
\item
if $\varphi \in \strongdiamondLTL$, then $\widetilde{\varphi} \in \strongdiamondLTL$; and
\item
if $\varphi \in \completeLTL$, then $\widetilde{\varphi} \in \completeLTL$.
\end{itemize}
\end{remark}

\subsection{The $\sfNEXPTIME$ upper bound for part~(1) of
Theorem~\ref{t: complexity LTL}}
\label{app: ss: upper bound of weak diamond LTL}

By Proposition~\ref{p: normal LTL},
we can assume that the input formula is always in normal form.
The proof of decidability itself is done by 
translating the input formula $\varphi \in \weakdiamondLTL$
to an equivalent ADC $(\AA_{\varphi},\CC_{\varphi})$.
The translation follows closely the classical translation
from standard LTL to B\"uchi automaton. 
(See, for example,~\cite{wolper-tutorial}.)
So we simply sketch it here. 
We recall the standard notion of the closure of the formula $\varphi$,
denoted by $\closure(\varphi)$.
\begin{itemize}
\item
$\varphi \in \closure(\varphi)$.
\item
$a \in \closure(\varphi)$, for each $a\in \Sigma$.
\item
If $\neg a \in \closure(\varphi)$,
then $\bigvee_{b \in \Sigma -\{a\}} b$.
\item
If $\varphi_1 \wedge \varphi_2 \in \closure(\varphi)$,
then $\varphi_1,\varphi_2 \in \closure(\varphi)$.
\item
If $\varphi_1 \vee \varphi_2 \in \closure(\varphi)$,
then $\varphi_1,\varphi_2 \in \closure(\varphi)$.
\item
If $\ttX \; \varphi_1 \in \closure(\varphi)$,
then $\varphi_1 \in \closure(\varphi)$.
\item
If $\varphi_1 \; \ttU \; \varphi_2 \in \closure(\varphi)$,
then $\varphi_1,\varphi_2 \in \closure(\varphi)$.
\item
If $\varphi_1 \; \ttR \; \varphi_2 \in \closure(\varphi)$,
then $\varphi_1,\varphi_2 \in \closure(\varphi)$.
\item
If $\Diamond^w \varphi_1 \in \closure(\varphi)$,
then $\varphi_1 \in \closure(\varphi)$.
\item
If $\neg \Diamond^w \varphi_1 \in \closure(\varphi)$,
then $\varphi_1 \in \closure(\varphi)$.
\end{itemize}
The standard construction of $\AA_{\varphi} = \langle Q, q_0, \mu, F\rangle$
will yield $Q \subseteq 2^{\sclosure(\varphi)}$, where $q \in Q$
if the conditions hold.
\begin{itemize}
\item[(S1)]
$\False \notin q$;
\item[(S2)]
$q \cap \Sigma$ is a singleton;
\item[(S3)]
if $\varphi_1 \in q$, then the normal form $\overline{\neg \varphi_1} \notin q$;
\item[(S4)]
if the normal form $\overline{\neg \varphi_1} \in q$, then $\varphi_1 \notin q$;
\item[(S5)]
if $\varphi_1 \wedge \varphi_2 \in q$, then $\varphi_1,\varphi_2 \in q$;
\item[(S6)]
if $\varphi_1 \vee \varphi_2 \in q$, then $\varphi_1\in q$ or $\varphi_2 \in q$.
\end{itemize}
Intuitively, the meaning of $\AA_{\varphi}$
is such that in every state $q$,
it takes care that every formula $q - \{\Diamond^w \psi \mid \psi \in \closure(\varphi)\}$ holds.
The construction of $q_0$, $\mu$ and $F$ are standard like in~\cite{wolper-tutorial}, thus, omitted.
The set $\CC$ of constraints will take care of the operator $\Diamond^w$.
It consists of the following.
\begin{enumerate}
\item
For every state $q$ that contains the sub-formula $\Diamond^w \psi$,
then $\CC_{\varphi}$ contains the constraints:
$$
V(q) \subseteq \bigvee_{\psi \in q'} V(q').
$$
\item
For every state $q$ that contains the sub-formula $\neg \Diamond^w \psi$,
then $\CC_{\varphi}$ contains the constraints:
$$
V(q) \cap V(q') = \emptyset,
$$
for all $q'$ that contains $\psi$.
\end{enumerate}
Now $\CC$ does not contain key-constraints.
The construction of $\AA_{\varphi}$ is already in $\sfEXPTIME$.
By $\sfNP$ upper bound in Theorem~\ref{t: emptiness diamond aut},
we get the $\sfNEXPTIME$ upper bound for
the satisfiability problem of $\weakdiamondLTL$.

\subsection{The 2-$\sfNEXPTIME$ upper bound for part~(2) of
Theorem~\ref{t: complexity LTL}}
\label{app: ss: upper bound of strong diamond LTL}

By Proposition~\ref{p: normal LTL},
we assume that the input formula $\varphi$ is in normal form.
Again, the proof of decidability is done by translating the input formula $\varphi \in \strongdiamondLTL$
to an equivalent ADC $(\AA_{\varphi},\CC_{\varphi})$.
However, we have to make a bit of modification because,
as explained in Example~\ref{eg: twice data values},
formulas such as $\ttG( a \to \Diamond^s a)$ cannot be directly translated to an ADC.

We apply the same trick as in Example~\ref{eg: twice data values}
to make a copy $\overline{\psi}$ of each formula $\psi$ in $\closure(\varphi)$.
The idea is that the copy $\overline{\psi}$
has exactly the same property as the formula $\psi$.
We denote by $\overline{\closure(\varphi)}$
the set of all such copies.

Then the ADC $(\AA_{\varphi},\CC_{\varphi})$ is defined 
over the alphabet $\Sigma \cup \overline{\Sigma}$ as follows.
The automaton $\AA_{\varphi} = \langle Q, q_0, \mu, F\rangle$ 
is such that $ Q \subseteq 2^{\sclosure(\varphi)\cup \overline{\sclosure(\varphi)}}$.
A state $q \in Q$ if in addition to Conditions~(S1)--(S6) above,
the following conditions hold.
\begin{itemize}
\item
both $\False$ and $\overline{\False}$ are not in $q$;
\item
$q \cap (\Sigma\cup \overline{\Sigma})$ is a singleton;
\item
if $\varphi_1 \in q$, then both $\neg \varphi_1$ and $\overline{\neg \varphi_1}$ are not in $q$;
\item
if $\neg \varphi_1 \in q$, then both $\varphi_1$ and $\overline{\varphi_1}$ are not in $q$;
\item
if $\varphi_1 \wedge \varphi_2 \in q$, then either
\begin{itemize}
\item
$\varphi_1,\varphi_2 \in q$, or
\item
$\overline{\varphi_1},\varphi_2 \in q$, or
\item
$\varphi_1,\overline{\varphi_2} \in q$, or
\item
$\overline{\varphi_1},\overline{\varphi_2} \in q$;
\end{itemize}

\item
if $\overline{\varphi_1 \wedge \varphi_2} \in q$, then either
\begin{itemize}
\item
$\varphi_1,\varphi_2 \in q$, or
\item
$\overline{\varphi_1},\varphi_2 \in q$, or
\item
$\varphi_1,\overline{\varphi_2} \in q$, or
\item
$\overline{\varphi_1},\overline{\varphi_2} \in q$;
\end{itemize}

\item
if $\varphi_1 \vee \varphi_2 \in q$, 
then one of $\varphi_1,\varphi_2,\overline{\varphi_1},\overline{\varphi_2}\in q$;
\item
if $\overline{\varphi_1 \vee \varphi_2} \in q$, 
then one of $\varphi_1,\varphi_2,\overline{\varphi_1},\overline{\varphi_2}\in q$;

\item
if $\overline{\varphi_1} \in q$, then $\varphi_1 \notin q$;
\item
if $\varphi_1 \in q$, then $\overline{\varphi_1} \in q$;

\item
if $\overline{\Diamond^s \varphi_1} \in q$,
then either $\Diamond^s \varphi \in q$, or $\Diamond^s \overline{\varphi_1} \in q$;

\item
if $\varphi_1$, then $\Diamond^s \varphi_1 \notin q$;
\item
if $\overline{\varphi_1}$, then $\Diamond^s \overline{\varphi_1} \notin q$.
\end{itemize}
Note that in such construction the states that are supposed to contain
both $\varphi_1$ and $\Diamond^s \varphi_1$ are replaced by
states that contain either 
\begin{itemize}
\item
both $\varphi_1$ and $\Diamond^s \overline{\varphi_1}$, or 
\item
both $\overline{\varphi_1}$ and $\Diamond^s \varphi$.  
\end{itemize}
The construction of $q_0$, $\mu$ and $F$ is standard.

The collection $\CC_{\varphi}$ of data-constraints consists of the following.
\begin{enumerate}
\item
For each state $q$ that contains both the sub-formulae $\psi$ and $\neg\Diamond^s \psi$,
$\CC_{\varphi}$ contains:
\begin{itemize}
\item
the key-constraints $V(q) \mapsto q$; and
\item
the denial-constraints $V(q) \cap V(p) = \emptyset$,
for all state $p$ that contains $\psi$.
\end{itemize}
The same if $q$ contains both $\overline{\psi}$ and $\neg \Diamond^s \overline{\psi}$
\item
For each state $q$ that contains the sub-formula $\Diamond^s \psi$ but not the sub-formula $\psi$,
$\CC_{\varphi}$ contains
the inclusion-constraints 
$$
V(q) \subseteq \bigvee_{\psi \in p} V(p).
$$
\item
For each state $q$ that contains the sub-formula $\neg \Diamond^s \psi$ but not the sub-formula $\psi$,
$\CC_{\varphi}$ contains the denial-constraints
$$
V(q) \cap V(p) = \emptyset,
$$
for every state $p$ that contains $\psi$.
\end{enumerate}

The construction of $\AA_{\varphi}$ is already in $\sfNEXPTIME$.
By Theorem~\ref{t: emptiness diamond aut},
we get the 2-$\sfNEXPTIME$ upper bound for $\strongdiamondLTL$.

\subsection{The 3-$\sfNEXPTIME$ upper bound for part~(3) of
Theorem~\ref{t: complexity LTL}}
\label{app: ss: upper bound of profile strong diamond LTL}

If we have the local comparison $\ttX_{\sim}$ and $\ttX_{\nsim}$,
it can be handled with the addition of profile in the automata.
As the inclusion of profile constraints
induce an exponential blow-up,
we get 3-$\sfNEXPTIME$ upper bound.
The construction is straightforward, thus, omitted.

\section{The $\sfNEXPTIME$-hardness of $\diamondLTL$} \label{app: s: hardness of weak diamond LTL}
In the proof of the following theorem we will use $\Box^w \varphi$ as an abbreviation of $\neg \Diamond^w \neg \varphi$. 
\begin{theorem}
The satisfiability problem for \weakdiamondLTL on (finite and infinite) data words is $\sfNEXPTIME$-hard.
\end{theorem}

\begin{proof}
The proof is by reduction from the $2^n$-corridor tiling problem. An instance 
$I = (T,H,V,F,L,n)$ of this problem consists of a finite set $T$ of tile types, horizontal and vertical constraints $H,V \subseteq T \times T$, constraints $F,L \subseteq T$ for the first and last row and a number $n$ given in unary. The task is to decide, whether $T$ tiles the 
$2^n \times 2^n$-corridor, respecting the constraints. This problem is $\sfNEXPTIME$-hard~\cite{chlebus-domino}.

For an arbitrary instance $I = (T,H,V,F,L,n)$ of the $2^n$-corridor tiling problem we will construct a formula $\varphi_I$ of polynomial length (in $|I|$) which is satisfiable if and only if $I$ has a solution.

We use $\Sigma = T \cup \{0,1\} \cup \{c,r,u\}$ as the underlying alphabet.
The idea is to assign to every square on the tiling grid a column and a row number to be able to check the constraints. We use data values as pointers to the binary encoding of a number. We first introduce some abbreviations.
A bit is represented by two successive positions in the data word. The first one is labelled by $0$ or $1$ and the data value of the second position serves as a pointer to the position with the next bit. It is crucial that all positions pointed by the same pointer carry the same bit value. The following formula ensures that this property holds. Since we will encode binary numbers with $n$ bits, the $\ttX$-operator is used $n-1$ times.    
\begin{align*}
\varphi_{bitstring} = & \Diamond^w((0 \vee 1) \wedge \ttX\Diamond^w( (0 \vee 1) \wedge \ttX\Diamond^w (\ldots \wedge \ttX\Diamond^w (0 \vee 1))\ldots )) \wedge \\ 
&(\Box^w 0 \vee \Box^w 1) \wedge \Box^w \ttX ((\Box^w 0 \vee \Box^w 1) \wedge \Box^w \ttX((\Box^w 0 \vee \Box^w 1) \wedge \ldots \wedge \Box^w \ttX(\Box^w 0 \vee \Box^w 1  )\ldots))
\end{align*}

The next formula encodes the number 0 in binary.
\begin{align*}
\varphi_{0} = \Box^w(0 \wedge \ttX\Box^w( 0 \wedge \ttX\Box^w (\ldots \wedge \ttX\Box^w 0)\ldots ))
\end{align*}
The following encodes number $2^n-1$.
\begin{align*}
\varphi_{1} = \Box^w(1 \wedge \ttX\Box^w( 1 \wedge \ttX\Box^w (\ldots \wedge \ttX\Box^w 1)\ldots ))
\end{align*}
The next formula says that the $i$th bit encodes bit value $b$. The expression $(\ttX\Box^w)^{i-1}$ means that $\ttX\Box^w$ is repeated $i-1$ times.
\begin{align*}
\bitat{i}{b} = \Box^w(\ttX\Box^w)^{i-1}b
\end{align*}
for $1\leq i \leq n$ and $b \in \{0,1\}$.
\\\\
It should be noted that the first bit serves as the lowest bit.

The formula $\varphi_I$ is composed of the formulas $\psi$ and $\chi$: the formula $\psi$ describes the encoding of the tiling grid and $\chi$ describes the constraints which has to hold.

Every square of the tiling is represented by a sequence of four positions in the data word. The first position is labeled by the tile type belonging to this square, the second one serves as a pointer to the bit representation of the column number of the square, the third one serves as a pointer to the bit representation of the row number of the square and the fourth one serves as an \textit{up-pointer} to the next upper square on the same column. Such a sequence of positions will be called \textit{square encoding}. 
\begin{align*}
\psi_1 &= G[\bigvee_{t \in T} t \rightarrow \ttX( c \wedge \varphi_{bitstring})] \\
\psi_2 &=  G[\bigvee_{t \in T} t \rightarrow \ttX\ttX( r \wedge \varphi_{bitstring})] \\
\psi_3 &=  G[\bigvee_{t \in T} t \rightarrow \ttX\ttX\ttX( u \wedge \Box^w \bigvee_{t \in T}t)] \\
\end{align*}
The first $4 \cdot 2^{2n}$ positions of the word represent a list of all square encodings of the tiling. The list begins with the square with column number $0$ and row number $0$. After all $2^n$ squares of a row are listed the first square of the next row follows. 

First we have to ensure that the first square encoding has row number 0 and column number 0.
\begin{align*}
\psi_4 = \ttX \varphi_0 \wedge \ttX\ttX \varphi_0
\end{align*}
A square encoding with column number $i < 2^n-1$ and row number $j$ is followed by a square encoding with column number $i+1$ and row number $j$.     
\begin{align*}
\psi_5  = & G[ (\bigvee_{t \in T} t \wedge \ttX \neg \varphi_1) \rightarrow (\bigwedge_{i=1}^n ((\ttX\ttX\bitat{i}{0} \rightarrow \ttX\ttX\ttX\ttX(\bigvee_{t \in T}t \wedge \ttX\ttX\bitat{i}{0}) \\ 
&\wedge (\ttX\ttX\bitat{i}{1} \rightarrow \ttX\ttX\ttX\ttX(\bigvee_{t \in T}t \wedge \ttX\ttX\bitat{i}{1}) ))] \\
\psi_6  = &G[ (\bigvee_{t \in T} t \wedge \ttX \neg \varphi_1) \rightarrow \bigvee_{i=1}^n (\bigwedge_{j=1}^{i-1}(\ttX\bitat{j}{1}\wedge \ttX\ttX\ttX\ttX\ttX\bitat{j}{0}) \wedge \\ 
&\ttX\bitat{i}{0} \wedge \ttX\ttX\ttX\ttX\ttX\bitat{i}{1} \wedge \\ 
&\bigwedge_{j=i+1}^{n}( (\ttX\bitat{j}{0}\rightarrow \ttX\ttX\ttX\ttX\ttX\bitat{j}{0}) \wedge (\ttX\bitat{j}{1}\rightarrow \ttX\ttX\ttX\ttX\ttX\bitat{j}{1}     )  ) )]
\end{align*}
A square encoding with column number $2^n-1$ and row number $j < 2^n-1$ is followed by a square encoding with column number $0$ and row number $j+1$.     
\begin{align*}
\psi_7 = &G[ (\bigvee_{t \in T} t \wedge \ttX \varphi_1 \wedge \ttX\ttX \neg \varphi_1) \rightarrow (\ttX\ttX\ttX\ttX(\bigvee_{t \in T}t \wedge \ttX\varphi_0))] \\
\psi_8 = &G[ (\bigvee_{t \in T} t \wedge \ttX \varphi_1 \wedge \ttX\ttX \neg \varphi_1) \rightarrow \bigvee_{i=1}^n (\bigwedge_{j=1}^{i-1}(\ttX\ttX\bitat{j}{1}\wedge \ttX\ttX\ttX\ttX\ttX\ttX\bitat{j}{0}) \wedge \\ 
&\ttX\ttX\bitat{i}{0} \wedge \ttX\ttX\ttX\ttX\ttX\ttX\bitat{i}{1} \wedge \\ 
&\bigwedge_{j=i+1}^{n}( (\ttX\ttX\bitat{j}{0}\rightarrow \ttX\ttX\ttX\ttX\ttX\ttX\bitat{j}{0}) \wedge (\ttX\ttX\bitat{j}{1}\rightarrow \ttX\ttX\ttX\ttX\ttX\ttX\bitat{j}{1}     )  ) )]
\end{align*}
After the square encoding with column number $2^n-1$ and row number $2^n-1$ there follow no more positions labelled with a tile type. By this we ensure that every square encoding occurs exactly once.     
\begin{align*}
\psi_9 = G[ (\bigvee_{t \in T} t \wedge \ttX \varphi_1 \wedge \ttX\ttX \varphi_1) \rightarrow \ttX G\bigwedge_{t \in T} \neg t ] 
\end{align*}
The up-pointer of every square encoding with column number $i$ and row number $j < 2^n-1$ points to the first position of the unique square encoding with column number $i$ and row number $j+1$.
\begin{align*}
\psi_{10} = &G[ (\bigvee_{t \in T} t \wedge \ttX\ttX\neg \varphi_1) \rightarrow (\bigwedge_{i=1}^n ((X\bitat{i}{0} \rightarrow \ttX\ttX\ttX\Box^w (\bigvee_{t \in T}t \wedge \ttX\bitat{i}{0}) \wedge \\ 
& (\ttX\bitat{i}{1} \rightarrow \ttX\ttX\ttX\Box^w (\bigvee_{t \in T}t \wedge \ttX\bitat{i}{1}) ))] \\
\psi_{11} = &G[ (\bigvee_{t \in T} t \wedge \ttX\ttX \neg \varphi_1) \rightarrow \bigvee_{i=1}^n (\bigwedge_{j=1}^{i-1}(\ttX\ttX\bitat{j}{1}\wedge \ttX\ttX\ttX\Box^w \ttX\ttX\bitat{j}{0}) \wedge \\ 
&\ttX\ttX\bitat{i}{0} \wedge \ttX\ttX\ttX\Box^w \ttX\ttX\bitat{i}{1} \wedge \\ 
&\bigwedge_{j=i+1}^{n}( (\ttX\ttX\bitat{j}{0}\rightarrow \ttX\ttX\ttX\Box^w \ttX\ttX\bitat{j}{0}) \wedge (\ttX\ttX\bitat{j}{1}\rightarrow \ttX\ttX\ttX\Box^w \ttX\ttX\bitat{j}{1}     )  ) )]
\end{align*}

The following formulas express that the constraints in $I$ are respected.
  
The squares of the first row carry only tile types from $F$.
\begin{align*}
\chi_1 = G[ (\bigvee_{t \in T} t \wedge \ttX\ttX \varphi_0) \rightarrow \bigvee_{t \in F}  t] 
\end{align*} 
Similarly, the squares of the last row carry only tile types from $L$. 
\begin{align*}
\chi_2 = G[ (\bigvee_{t \in T} t \wedge \ttX\ttX \varphi_1) \rightarrow \bigvee_{t \in L}  t] 
\end{align*} 

The tile type of a square and the tile type of his right neighbor respect the horizontal constraints.     
\begin{align*}
\chi_3 = G[ \bigwedge_{t \in T} ((t \wedge \ttX\neg \varphi_1) \rightarrow \ttX\ttX\ttX\ttX \bigvee_{(t,t')\in H} t') ]
\end{align*} 
The tile type of a square and the tile type of his upper neighbor respect the vertical constraints.     
\begin{align*}
\chi_4 = G[ \bigwedge_{t \in T} ((t \wedge \ttX\ttX\neg \varphi_1) \rightarrow \ttX\ttX\ttX\Box^w \bigvee_{(t,t')\in V} t') ]
\end{align*} 

The desired formula is $\varphi_I = \bigwedge_{i=1}^{11} \psi_i \wedge \bigwedge_{j=1}^{4} \chi_j $. It's easy to see that $\varphi_I$ is satisfiable if and only if $I$ has a solution.
\end{proof}

\end{document}

%% file: commands.tex
\newtheorem{claim}{Claim}
\newtheorem{eg}{Example}

\newcommand {\bitat}[2]{\mbox{\textit{bit-#1-#2}}}

\newcommand{\completeLTL}{\mbox{$\textrm{LTL}[\Diamond^{w},\Diamond^{s},\ttX_{\sim},\ttX_{\nsim}]$}}
\newcommand{\completeLTLstrong}{\mbox{$\textrm{LTL}[\Diamond^{s},\ttX_{\sim},\ttX_{\nsim}]$}}
\newcommand{\diamondLTL}{\mbox{$\textrm{LTL}[\Diamond]$}}
\newcommand{\weakdiamondLTL}{\mbox{$\textrm{LTL}[\Diamond^{w}]$}}
\newcommand{\strongdiamondLTL}{\mbox{$\textrm{LTL}[\Diamond^{s}]$}}

\newcommand{\fD}{\mbox{$\mathfrak{D}$}}

\newcommand{\nn}{\mbox{${\mathbb N}$}}

\newcommand{\ttX}{\mbox{$\texttt{X}$}}

\newcommand{\ttU}{\mbox{$\texttt{U}$}}

\newcommand{\ttR}{\mbox{$\texttt{R}$}}
\newcommand{\ttF}{\mbox{$\texttt{F}$}}
\newcommand{\ttG}{\mbox{$\texttt{G}$}}

\newcommand{\LL}{{\cal L}}

\renewcommand{\AA}{\mbox{$\mathcal{A}$}}
\newcommand{\CC}{\mbox{$\mathcal{C}$}}

\renewcommand{\LL}{\mbox{$\mathcal{L}$}}
\newcommand{\MM}{\mbox{$\mathcal{M}$}}

\renewcommand{\SS}{\mbox{$\mathcal{S}$}}
\newcommand{\ZZ}{\mbox{$\mathcal{Z}$}}

\newcommand{\sCC}{\mbox{\scriptsize $\mathcal{C}$}}

\newcommand{\sSS}{\mbox{\scriptsize $\mathcal{S}$}}

\newcommand{\True}{\mbox{$\textsf{True}$}}
\newcommand{\False}{\mbox{$\textsf{False}$}}

\newcommand{\closure}{\mbox{$\textsf{Cl}$}}
\newcommand{\sclosure}{\mbox{\tiny $\textsf{Cl}$}}

\newcommand{\sfNP}{\mbox{$\textsf{NP}$}}
\newcommand{\sfNEXPTIME}{\mbox{$\textsf{NEXPTime}$}}
\newcommand{\sfEXPTIME}{\mbox{$\textsf{EXPTime}$}}
\newcommand{\sfEXPSPACE}{\mbox{$\textsf{EXPSpace}$}}
\newcommand{\sfPSPACE}{\mbox{$\textsf{PSpace}$}}

\newcommand{\sfProj}{\mbox{$\textsf{Proj}$}}
\newcommand{\sfInf}{\mbox{$\textsf{Inf}$}}
\newcommand{\sfProfile}{{\sf Profile}}
\newcommand{\sfParikh}{\mbox{$\textsf{Parikh}$}}
\newcommand{\sfImage}{\mbox{$\textsf{Image}$}}
\newcommand{\sfprofile}{\mbox{$\textsf{profile}$}}

\newcommand{\tsfImage}{\mbox{\tiny $\textsf{Image}$}}

\newcommand{\sfZonal}{\mbox{\sf Zonal}}
\newcommand{\tsfZonal}{\mbox{\sf {\tiny Zonal}}}
\newcommand{\tscZonal}{\mbox{$\textsc{\tiny zonal}$}}

\newcommand{\tsffin}{\mbox{\tiny $\textsf{fin}$}}

\newcommand{\sfFO}{\mbox{$\textsf{FO}$}}
\newcommand{\sfMSO}{\mbox{$\textsf{MSO}$}}
\newcommand{\sfEMSO}{\mbox{$\exists\textsf{MSO}$}}

\newcommand{\leaveout}[1]{}
